\documentclass[12pt]{article}
\usepackage{graphicx}
\usepackage{hyperref}
\usepackage{color}
\usepackage{float}

\usepackage{makeidx}
\usepackage{graphicx}

\usepackage{multicol}
\usepackage[centertags]{amsmath}
\usepackage{amsfonts}
\usepackage{euscript}
\usepackage{subfigure}
\usepackage{color}
\usepackage{amsmath,amssymb,amsthm}

\def\vect#1{\mbox{\boldmath{$#1$}}}

\def\Ac{{\cal A}}
\def\Bc{{\cal B}}

\def\vx{\vec{\vect x}}
\def\vy{\vec{\vect y}}
\def\vz{{ \vec{\vect z}}}

\def\mC{\mathbb{C}}
\def\Cc{{\cal C}}
\newcommand{\cC}{\mathcal{C}}
\newcommand{\cA}{{\mathcal{A}}}
\def\Dc{{\cal D}}

\def\wG{G}

\newcommand{\bx}{\vect \rho}
\newcommand{\bz}{\vect z}

\def\wg{g}

\newcommand{\argmin}{\mathop{\mbox{argmin}}}

\newcommand{\eps}{\varepsilon}
\renewcommand{\epsilon}{\eps}
\renewcommand{\leq}{\leqslant}
\renewcommand{\geq}{\geqslant}

\def\bfrho{\mbox{\boldmath$\rho$}}
\def\bfb{\mbox{\boldmath$b$}}

\newcommand{\supp}{\mathrm{supp}}

\newtheorem{lemma}{Lemma}
\newtheorem{remark}{Remark}
\newtheorem{proposition}{Proposition}%[section]
\newtheorem{definition}{Definition}%[section]

\usepackage{perpage}
\MakePerPage{footnote}

\begin{document}

\title{Imaging with  highly incomplete and corrupted data}
%\title[Imaging sparse scenes from noisy data]{Imaging sparse scenes from noisy data}

\author{Miguel Moscoso, Alexei Novikov, George Papanicolaou, Chrysoula Tsogka}

\maketitle
%\maketitle
%%%%%%%%%%%%%%%%%%%%%%%%%%%%%
\begin{abstract}
We consider  the problem of imaging sparse scenes from a few noisy data using an $\ell_1$-minimization approach. This problem can be cast as 
a linear system of the form $\Ac \, \bfrho =\bfb$, where $\Ac$ is an $N\times K$ measurement matrix. We assume that the dimension of the unknown
sparse vector $\bfrho \in \mC^K$ is much larger than the dimension of the data vector $\bfb \in \mC^N$, i.e, $K \gg N$.
We provide a theoretical framework that allows us to examine under what conditions the $\ell_1$-minimization problem admits a solution 
that is close to the exact one in the presence of noise.  Our analysis shows that $\ell_1$-minimization is not robust for imaging with noisy data when high
resolution is required. 
To improve the performance of  $\ell_1$-minimization we propose to solve instead 
the augmented linear system $ [\Ac \, | \,  \Cc] \bfrho =\bfb$, where the $N \times \Sigma$ matrix $\Cc$ is a noise collector. It is constructed so as its column vectors provide a frame 
on which the noise of the data, a vector of dimension $N$, can be well approximated.  Theoretically, the dimension $\Sigma$ of the noise collector should be $e^N$ 
which would make its use not practical. However, our numerical results illustrate that 
%$\ell_1$-minimization is stabilized and 
robust results in the presence of 
noise can be obtained with a large enough number of columns $\Sigma \approx 10 K$.    
\end{abstract}
\vspace{2pc}

\maketitle
%******************************************************************
\section{Introduction}
%******************************************************************
In this paper, we are interested in imaging problems formulated as 
\begin{equation}\label{family0intro}
\Ac \, \bfrho=\bfb\, ,
\end{equation}
so the data vector $\bfb \in\mC^N$ is a linear transformation of the  unknown vector $\bfrho \in\mC^K$ that represents the image. 
The model matrix $\Ac\in \mC^{N\times K}$, which  is given to us, depends
on the geometry of the imaging system and on the sought resolution.
Typically, the linear system (\ref{family0intro}) is underdetermined because only a few linear measurements are gathered,
%the number of  unknowns $K$ is much larger than the number of measurements $N$,  
so $N \ll K$. Hence, there exist infinitely many  solutions to (\ref{family0intro}) and, thus, it is a priori not possible to find the 
correct one  without some additional information.

We are interested, however, in imaging problems with sparse scenes. We seek to locate the positions and amplitudes of a small number $M$ of point sources that illuminates a linear array of detectors. This means
that the unknown vector $\bfrho$ is M-sparse, with only a few $M \ll K$ non-zero entries. Under this assumption, (\ref{family0intro}) falls under 
the compressive sensing framework \cite{Gorodnitsky97,Donoho03,Gribonval03,Candes06a}. It follows from \cite{Donoho03} 
that the unique M-sparse solution of  (\ref{family0intro}) can be obtained with $\ell_1$-norm minimization when the model matrix 
$\Ac$ is  incoherent, i.e., when its mutual coherence\footnote{The mutual coherence of $\Ac$ is defined as $\max_{i \ne j}  
|\langle \vect a_i,\vect a_j \rangle |$, where the column vectors $\vect a_i \in \mC^N$ of $\Ac$ are normalized to one, so that $\|\vect a_i \|_{\ell_2} =1$ $\forall \, i=1,\ldots,K$.}  
is smaller than $1/(2M)$. The same result can be obtained assuming $\Ac$ obeys the M-restricted isometry property  
\cite{Candes06a}, which basically states that all sets of M-columns of $\Ac$ behave approximately as an orthonormal system.  

In our imaging problems these incoherence conditions 
can be satisfied only for coarse image discretizations that imply poor resolution.  
To retain resolution and recover the position of the sources with higher precision we 
propose to extend the theory so as to allow for some coherence in $\Ac$. To this end, we show that uniqueness for the minimal $\ell_1$-norm solution of (\ref{family0intro}) can 
be obtained under less restrictive conditions on the model matrix $\Ac$. More specifically, given  
the columns of $\Ac$  that correspond to the support of $\bfrho$,  we define their vicinities as the sets of columns that are almost parallel
\footnote{The vicinity of a column $\vect a_i$ is defined as the set of all columns $\vect a_j$ such that 
$|\langle \vect a_i,\vect a_j \rangle | \ge 1/(3M)$.} 
to them. With this definition, our first result  set out in Proposition \ref{Old_l1} states that if the sources are  located far enough 
from each other, so that their vicinities do not overlap,
we can recover their positions exactly with noise free data. 
Furthermore, in the presence of small noise, their position is still approximately recoverable, in the sense that most of the solution vector
is supported in the vicinities while some small noise (grass) is present away from them.

This result finds interesting applications in imaging. As we explain in Section~\ref{sec:array}, in array imaging we seek to find the position of point sources that are represented as the non-zero entries of  $\bfrho$. Our result states under what conditions the location of these objects can 
be determined with high precision. It can be also used to explain super-resolution, i.e.,   
the significantly superior resolution that $\ell_1$-norm minimization provides compared to the conventional resolution 
of the imaging system, i.e., the Rayleigh resolution. 
For instance, super-resolution have been studied using {sparsity} promotion for sparse spike trains  recovery from band-limited measurements. Donoho \cite{Donoho92} showed that spike locations and their weights can be exactly recovered for a cutoff frequency $f_c$ if the minimum spacing $\Delta$ between spikes is large enough, so $\Delta>1/f_c$. Cand\`{e}s and Fernandez-Granda \cite{Candes14} showed that $\ell_1$-norm minimization guarantees the exact 
recovery if $\Delta>1/2f_c$.
Super-resolution has also been studied for highly coherent model matrices $\Ac$ that arise in imaging under the assumption  of well separated objects  when the resolution is below the Rayleigh threshold  \cite{Fannjiang12b, Borcea15, Borcea18}. These works 
%\cite{Fannjiang12b, Borcea15, Borcea18} 
include results regarding the robustness of super-resolution in the presence of noise.

Our theory also addresses the robustness to noise of the minimal $\ell_1$-norm solution. Specifically, we show that for noisy data the solution $\bfrho$ can 
be separated into two parts: (1) the coherent part which is supported inside the vicinities, and (2) the incoherent part, usually referred 
to as grass, that is small and it is present everywhere.  A key observation of our work is that the $\ell_1$-images get worse as $\sqrt{N}$ when there is noise in the data and, thus, 
$\ell_1$-norm  minimization fails when the number of measurements $N$ is large. This basically follows from (\ref{def:gamma}) in Proposition \ref{Old_l1} which 
relates the $\ell_1$ norm of the solution to the $\ell_2$ norm of the data, so 
 $$\| \vect \rho \|_{\ell_1} \leq  \gamma \,  \| \vect b \|_{\ell_2} .$$
The key quantity here is the constant $\gamma$, which for usual imaging matrices $\cA$ is proportional to $\sqrt{N}$. 

To overcome this problem we introduce in Proposition \ref{noise_c} the noise collector matrix $\cC \in \mC^{N\times \Sigma}$ and propose to solve instead the augmented linear system 
$[\Ac \, | \,  \Cc]\bfrho=\bfb$. The dimension of the unknown vector $\bfrho$ is, thus, augmented by $\Sigma$ components which do not have any physical meaning. They correspond to a fictitious sources that 
allows us to better approximate the noisy data.  The natural question is how to build the noise collector matrix. Theoretically, the answer is given in the proof of 
Proposition \ref{noise_c}  in Section~\ref{sec:la},
which is constructive. The key is that the column vectors of $[\Ac \, | \,  \Cc]$ form now a frame in which the noisy vector $\bfb$ can be well approximated. As a consequence, we obtain a bound on the 
constant $\gamma$ ($\gamma < 18 M^2$) which is now independent of $N$. The drawback of this construction is that we need exponentially many vectors, that is $\Sigma \lesssim e^N$. This would 
suggest that the noise collector may not be practical.  However, the numerical experiments show that with a large enough number of columns in $\cC$ selected at random 
(as i.i.d. Gaussian random variables with mean zero and variance $1/N$) the $\ell_1$-norm minimization problem is regularized and the minimal $\ell_1$-norm solution is found.

The paper is organized as follows. In Section \ref{sec:array}, we formulate the array 
 imaging problem.  In Section \ref{sec:la}, we present in a abstract linear algebra framework the conditions under 
 which $\ell_1$-minimization provides the exact solution to problem (\ref{family0intro}) with and without noise.  This section contains our main results.
  In Section \ref{sec:vicinities},  we illustrate with numerical simulations 
how our abstract theoretical results are relevant in imaging sparse sources with noisy data. Section \ref{sec:conclusions} contains our conclusions.

\section{Passive array imaging}
   \label{sec:array}
    \begin{figure}[t]
    \begin{center}
    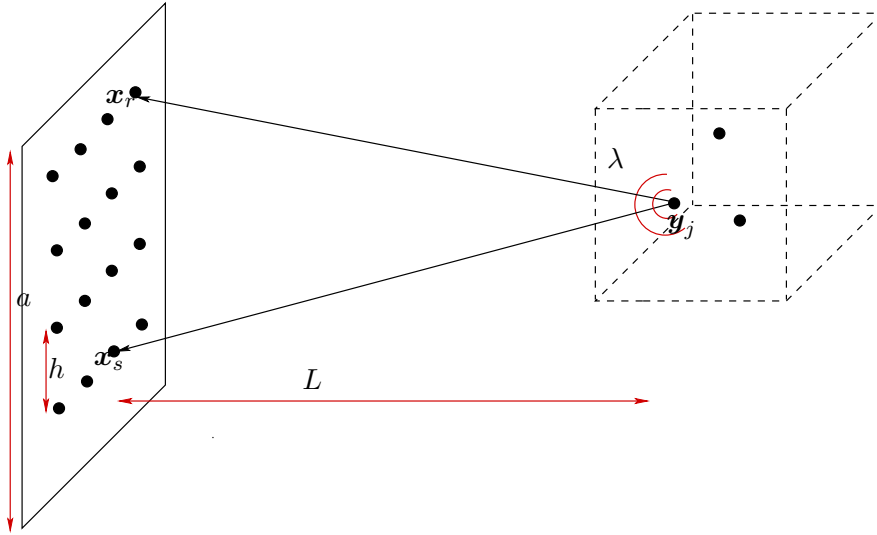
        \end{center}
        \caption{General setup of a passive array imaging problem. The source at $\vx_s$  emits a signal and the response is recorded at all array elements $\vx_r$, $r=1,\ldots,N$. The sources located at ${\vy}_{j}$, $j=1,\dots,M$ 
    are at distance $L$ from the array and inside the image window IW. }
    \label{cs_3d_illustration}
    \end{figure}

We consider point sources located inside a region of interest called the image window IW. 
The goal of array imaging is to determine their positions and amplitudes using measurements obtained on an array of receivers.
The array of size $a$ has $N$ receivers separated by a distance $h$  
%It has a characteristic length $a$  (see Figure \ref{cs_3d_illustration}). The receivers are 
located at positions $\vx_r$, $r = 1,\ldots,N$ (see Figure \ref{cs_3d_illustration}).  
They can measure single or multifrequency signals with frequencies $\omega_l$, $l=1,\dots,S$. 
The $M$ point sources, whose positions $\vz_j$ 
and complex-valued amplitudes $\alpha_j\in\mC$, $j=1,\dots,M$, we seek to determine, are at a distance $L$ from the array. The ambient medium between the array and the sources can be homogeneous or inhomogeneous. 

In order to form the images we discretize the IW using a uniform grid of points $\vy_k$, $k=1,\ldots,K$, and we introduce the {\it  true 
source vector}
$$\vect \rho=[\rho_{1},\ldots,\rho_{K}]^{\intercal}\in\mC^K\, ,$$ 
such that
\[
\rho_{k} = \left\{ 
\begin{array}{ll}
\alpha_j, & \hbox{ if } \| \vz_j - \vy_{k} \|_\infty <   \hbox{ grid-size, for some } j=1,\ldots,M,\\
0, & \hbox{ otherwise.}
\end{array}
\right. 
\]
We will not assume that the sources
lie on the grid, i.e., typically $\vz_j  \neq \vy_k$ for all $j$ and $k$. 
To write the data received on the array in a compact form, we define the Green's function vector   
\begin{equation}\label{eq:GreenFuncVec}
\vect \wg(\vy;\omega)=[\wG(\vx_{1},\vy;\omega), \wG(\vx_{2},\vy;\omega),\ldots,
\wG(\vx_{N},\vy;\omega)]^{\intercal}\,
\end{equation}
at location $\vy$ in the IW, where $\wG(\vx,\vy;\omega)$ denotes the free-space Green's function of
the homogeneous medium. This function characterizes the propagation of a 
signal of angular frequency $\omega$ from point $\vy$ to point $\vx$, so 
(\ref{eq:GreenFuncVec}) 
represents the signal received at the array due to a point source of amplitude one, phase zero, and frequency $\omega$
at $\vy$. If the medium is homogeneous
\begin{equation}\label{greenfunc}
\wG(\vx,\vy;\omega)=\frac{\exp \left( i \frac{\omega |\vx-\vy|}{c_0}  \right)}{4\pi|\vx-\vy|} .
\end{equation}
The signal received at $\vx_r$ at  frequency $\omega_l$  is given by 
\begin{equation}
\label{response}
b(\vx_r, \omega_l)=\sum_{j=1}^M\alpha_j  \wG (\vx_r,\vz_j; \omega_l).
\end{equation}
If we  normalize the columns of $\Ac$ to one and stack the data in a column vector 
\begin{equation}
\vect b = \frac{1}{\sqrt{N S}} [ b(\vx_1, \omega_1), b(\vx_2, \omega_1),\dots,b(\vx_N, \omega_S)]^{\intercal}\,,
\end{equation}
then the source vector $\bfrho$ solves the system  
$\Ac\,\bfrho = \bfb$,
with the $(N\cdot S)\times K$ matrix
\begin{equation}\label{eq:matrixSfreq}
\hspace*{-1.5cm}  \Ac
 = \frac{1}{\sqrt{N S}}
\left( %\begin{pmatrix}
\begin{array}{cccc}
    \uparrow & \uparrow&        & \uparrow \\
    \vect\wg(\vy_1; \omega_1)   & 
    \vect\wg(\vy_2; \omega_1)    & 
    \ldots & 
     \vect\wg(\vy_K; \omega_1)    \\
\downarrow & \downarrow&        & \downarrow \\
    \uparrow & \uparrow&        & \uparrow \\
    \vect\wg(\vy_1; \omega_2)   & 
    \vect\wg(\vy_2; \omega_2)    & 
    \ldots & 
  \vect\wg(\vy_K; \omega_2)    \\
\downarrow & \downarrow&        & \downarrow \\
\vdots & \vdots&        & \vdots \\
    \uparrow & \uparrow&        & \uparrow \\
\vect\wg(\vy_1; \omega_{S})   & 
\vect\wg(\vy_2; \omega_{S})    & 
    \ldots & 
  \vect\wg(\vy_K; \omega_{S})    \\
\downarrow & \downarrow&        & \downarrow
  \end{array}
\right) %\end{pmatrix} 
: = 
\left(
\begin{array}{cccc}
    \uparrow & \uparrow&        & \uparrow \\
    \vect a_1   & 
    \vect a_2& 
    .. & 
    \vect a_K    \\
\downarrow & \downarrow&        & \downarrow
  \end{array}
\right)  .
 \end{equation}
The system $\Ac\,\bfrho = \bfb$ relates the unknown vector  $\bfrho\in\mC^K$ to the data vector $\bfb \in\mC^{(N\cdot S)}$.
This system of linear equations  can be solved by appropriate $\ell_2$
and $\ell_1$ methods. 

\begin{remark}
For simplicity of the presentation, we restricted ourselves to the passive array imaging problem where we seek to determine a distribution of sources. The active array imaging problem can be cast under the same linear algebra framework assuming  the linearized Born approximation for scattering \cite{CMP13}. In that case, we still obtain a system of the form $\Ac_s \,\bfrho = \bfb$, where $\bfrho$ is the reflectivity of the scatterers, $\bfb$ is the data, and $\Ac_s$ is a model matrix for the scattering problem defined in a similar manner to (\ref{eq:matrixSfreq}). Even more, when multiple scattering
is not negligible the problem can also be cast as in (\ref{family0intro}); see \cite{CMP14} for details. Therefore, the theory presented in the next sections can be applied to the scattering problems provided that the matrix $\Ac_s$ satisfies the assumptions of Propositions \ref{Old_l1} and \ref{noise_c}. 
\end{remark}

\section{$\ell_1$ minimization-based methods}\label{sec:la}

In the imaging problems considered here we assume that the sources occupy only a small fraction of the image window IW. This means that the true source vector $\bfrho$ is sparse, so the 
number of its entries that are different than zero, denoted by $M$, is much smaller than its length $K$. Thus, we assume $M=|\supp(\bfrho)|\ll K$. This prior knowledge changes the imaging problem substantially because
 we can exploit the sparsity of $\bfrho$ by formulating it %$\Ac\,\bfrho = \bfb$ 
 as an optimization problem which seeks the sparsest vector in $\mC^K$ that equates model and data. 
 Thus, for a single measurement vector $\vect b$ we solve
\begin{equation}\label{l1normsol}
\bfrho_{\ell_1}=\argmin \|\vect \rho\|_{\ell_1}, \hbox{ subject to } \Ac \vect \rho= \vect b.
\end{equation}
Above, and in the sequel, we denote by $\| \cdot \|_{\ell_2}$, $\| \cdot \|_{\ell_1}$  and  $\| \cdot \|_{\ell_\infty}$ the $\ell_2$, $\ell_1$  and  $\ell_\infty$ norms of a vector, respectively. 

 The $\ell_1$ minimization problem~(\ref{l1normsol}) can be solved efficiently in practice. 
Several methods
have been proposed in the literature to solve it.
Here are some of
them: orthogonal matching pursuit~\cite{Bruckstein09}, 
homotopy~\cite{Tibshirani94,Osborne00, Efron04}, interior-point methods~\cite{Alizadeh95,Wright05},
gradient
projection~\cite{Figueiredo07}, sub-gradient descent methods in primal and dual spaces~\cite{Nedic01,Borwein11},
and proximal gradient in combination with
iterative shrinkage-thresholding~\cite{Nesterov83, Nesterov07,
Beck09}.  In this work we use GeLMA~\cite{Moscoso12},  a semi-implicit version 
of the primal-dual method~\cite{CP} that converges to the solution of the following problem independently of the regularization
parameter $\tau$: define the function
\begin{equation}\label{min-max}
F(\bx, \bz) = \tau \| \bx \|_{\ell_1} + \frac{1}{2} \| \cA  \bx  - \bfb \|^2_{\ell_2} + \langle \bz, \bfb - \cA \bx \rangle
\end{equation}
for $\bx \in \cC^K$ and $\bz \in \cC^N$,
and determine the solution as
$$  \max_{\bz} \min_{\bx} F(\bx,\bz) .$$

In the literature of compressive sensing we find the following theoretical justification of the $\ell_1$-norm minimization approach.
If we assume decoherence of the columns  of $\Ac$, so
\begin{equation}\label{ortho}
|\langle \vect a_i, \vect a_j\rangle |  <  \frac{1}{2 M},\quad \forall i \neq j, 
\end{equation}
 then the $M$-sparse solution of
$\Ac \bx= \vect b$ is unique, and it
can be found as the solution of~(\ref{l1normsol})~\cite{Gorodnitsky97,Donoho03,Gribonval03}.
Numerically, the $\ell_1$-norm minimization approach works under less restrictive conditions than  the decoherence  condition~(\ref{ortho}) suggests.
In fact, our imaging matrices almost never satisfy (\ref{ortho}). 

Consider  a typical imaging regime with central wavelength $\lambda_0$.   
Assume we use $S=36$  equally spaced frequencies covering a bandwidth that is $10\%$ of the central frequency. The size of the array is $a$ %=100\lambda_0$, 
and the distance between the array and the IW is $L=a$. 
An IW of size is $30 \lambda_0 \times 30 \lambda_0$  is discretized
using a uniform grid with mesh size $\lambda_0/2 \times \lambda_0/2$. For such parameters, {\it every} column  vector $ \vect a_i$ has 
 at least {\it sixty two} other column vectors $\vect a_j$ so that   $|\langle \vect a_i, \vect a_j\rangle | \geq 1/16$. 
Thus, our matrices are fairly far from satisfying the decoherence  condition~(\ref{ortho}) if we want to recover, say, 8 sources. Numerically, however, the $\ell_1$ minimization
works flawlessly. 

Physically, a pair of columns  $\vect a_i$ and $\vect a_j$ are  {\em coherent}, so $|\langle \vect a_i, \vect a_j \rangle| \approx 1$, if the corresponding grid-points in the image are close to each other. In other words, when  $ \vect a_i$ lies in a vicinity of  $\vect a_j$ (and vice versa). We assume, though, that the sources are far apart and, thus, the  
%corresponding columns of $\Ac$ are approximately orthogonal.  Thus, any 
the set of columns indexed by the support of the true source vector $\bfrho$ does satisfy the the decoherence condition (\ref{ortho}). The above observation motivates the following natural conjecture. Perhaps, the $\ell_1$ minimization works well because it suffices to satisfy (\ref{ortho}) only on the support of $\bfrho$. Our main result supports this conjecture.

\subsection{Main results}
 When data is perturbed by small noise, the following qualitative description of the image could be observed. Firstly, 
 some pixels close to the points where the sources are located become visible. 
Secondly, a few pixels away from the sources are also visible. The latter is usually referred as grass. In order to quantify the observed results we 
need to modify the  the decoherence condition~(\ref{ortho}) and introduce the vicinities.

\begin{definition}\label{DeCo}
Let $\vect \rho \in \mathbb{C}^K$ be an $M$-sparse solution of
$\Ac \vect \rho=\vect b$, with support $T=\{i: \rho_i\neq 0\}$  \footnote{Below and in the rest of the paper the notation $\rho_{i}$ means the $i$th entry of the vector $\vect \rho$. In contrast,  we use the notation $\vect \rho_i$ 
to represent the $i$th  vector of a set of vectors.}.  For any $j \in T$ define the corresponding vicinity of $\vect a_j$ as 
\begin{equation}\label{vicinity}
%S_j=\left\{ k ~ \,\, \mbox{s.t.} \,\, ~ |\langle \vect a_k, \vect a_j \rangle | \geq \frac{1}{3 M}\right\}.
S_j=\left\{ k ~ :  ~ |\langle \vect a_k, \vect a_j \rangle | \geq \frac{1}{3 M}\right\}.
\end{equation} For any vector  $\vect \eta \in \mathbb{C}^K$ its coherent misfit to $\vect \bfrho$ is
\begin{equation}\label{co_}
{\bf Co}(\vect \rho, \vect \eta)= \sum_{j \in T} \left| \rho_j - \sum_{k \in S_j} \langle \vect a_j, \vect a_k \rangle \eta_k \right|, 
\end{equation}
whereas its incoherent remainder with respect to $\vect \bfrho$ is
\begin{equation}\label{in_}
{\bf In}(\vect \rho, \vect \eta)=  \sum_{k \not\in \Upsilon} |\eta_k |, ~ \,\, ~ \Upsilon=  \cup_{j \in T} S_j. 
\end{equation}
\end{definition}

\begin{proposition} \label{Old_l1}
Let $\vect \rho$ be an $M$-sparse solution of
$\Ac \vect \rho=\vect b$, and let $T$ be its support. 
Suppose  the vicinities  $S_j$ from Definition~\ref{DeCo}  do not overlap,
and let $\gamma > 0$ be defined as
\begin{equation}
 \gamma = \sup_{\vect c} \frac{ \|  \vect \xi \|_{\ell_1}}{\| \vect c \|_{\ell_2}},  \mbox{ where }   \vect \xi  \mbox{ is the minimal } \ell_1-\mbox{norm} \,\, \mbox{solution of }\cA  \, \vect \xi =\vect c. 
\label{def:gamma}
\end{equation}
%\begin{equation}\label{as:sine-}
%\forall \vect b\, \exists \vect \rho \hbox{ such that }  \Ac \vect \rho=\vect b \hbox{ and } \| \vect \rho \|_{\ell_1} \leq  \gamma  \| \vect b \|_{\ell_2} 
%\end{equation}
Let $\vect \rho_\delta$ be the minimal $\ell_1$-norm  solution of the noisy problem 
\begin{equation} \label{sys_noise}
\min  \|\vect \rho_\delta \|_{\ell_1}, \hbox{ subject to } \Ac \vect \rho_\delta =  \vect b_\delta,
\end{equation}
with $ \| \vect b - \vect b_\delta  \|_{\ell_2} \leq \delta$. Then,
\begin{equation}\label{est_co-}
{\bf Co}(\vect \rho, \vect \rho_\delta) \leq  3 \gamma \delta,
\end{equation}
and
\begin{equation}\label{est_in-}
{\bf In}(\vect \rho, \vect \rho_\delta) \leq  5 \gamma  \delta.
\end{equation}
If  $\delta=0$, and $\Upsilon$  does not contain collinear vectors, we have exact recovery: $\vect \rho_\delta= \vect \rho.$ 
\end{proposition}

Proposition~\ref{Old_l1} is proved in \ref{uno}. As it follows from this proof, %of Proposition \ref{Old_l1}, 
our pessimistic bound $1/(3M)$ 
could be sharpened to the usual bound (\ref{ortho})  found in the literature. We did not strive to 
obtain sharper results because it will make the proofs more technical and, more importantly,  because the concept of vicinities describes well the observed phenomena in imaging with this bound. 

When there is no noise so $\delta = 0$, Proposition~\ref{Old_l1}  tells us that the M-sparse solution of $\Ac \vect \rho =\vect b$ 
can be recovered exactly by solving the $\ell_1$ minimization problem under a less stringent condition than (\ref{ortho}). 
Note that we allow for the columns of $\Ac$ to be close to collinear. When there is  noise so $\delta \neq 0$, this Proposition
shows that if the data $\vect b$ is not exact but it is known up to some bounded vector,  the solution 
$\vect \rho_\delta$ of the minimization problem (\ref{sys_noise}) is close to the solution of the original (noiseless) problem in the following sense. The 
solution $\vect \rho_\delta$  can be separated into two parts: the coherent part supported in the vicinities $S_j$ of the true solution, $j \in T$, 
and the incoherent part, which is small for low noise, and %usually referred to as grass in imaging. The grass 
that is supported away from these vicinities. % and it is shown to be small assuming that the $\ell_2$-norm of the noise  is small.
Other stability results can be found in \cite{Candes06a,Candes06b,DET06,TROPP06-1,Fannjiang12b,Borcea15}.
  
Let us now make some comments regarding the relevance of this result in imaging.  
Vicinities, as defined in (\ref{vicinity}), are related to the classical $\ell_2$-norm resolution theory. Indeed, recall Kirchhoff 
migration imaging given by the $\ell_2$-norm solution
\begin{equation} \label{sol_ell2}
\vect \rho_{\ell_2} = \Ac^* \vect \bfb,
\end{equation}
where $\Ac^*$ is the conjugate transpose of $\Ac$.  Since $\bfb = \Ac \bfrho$ the resolution analysis of KM relies on studying the behaviour of the inner products
$ |\langle \vect a_i, \vect a_k \rangle|$. We know from classical resolution analysis \cite{Borcea03} that the inner 
products $ |\langle \vect a_i, \vect a_k \rangle| $ are large for points $\vect y_k$ that fall inside the support of the  KM point spread 
function, whose size is  $\lambda L/a$ in cross-range (parallel to the array) and $c/B$ in range  (perpendicular to the array). 
Hence, we expect the size of the vicinities to be proportional to these classical resolution limits, with an appropriate scaling factor that is inversely 
proportional to the sparsity $M$. This is illustrated with numerical simulations in Section \ref{sec:vicinities} (right column of Figure  \ref{fig:th1}). 

Under this perspective, one could argue that Proposition~\ref{Old_l1} tells us the well known result that a good reconstruction can be obtained for well separated sources. 
Proposition~\ref{Old_l1}, however, gives us more information, it provides an $\ell_1$-norm resolution theory for imaging: when vicinities do not overlap, there is a single non-zero element of the source associated within each vicinity. Permitting the columns of $\Ac$ to be almost collinear inside the vicinities allows for a fine discretization inside the vicinities and therefore the 
source can be recovered with very high precision. Furthermore, recovery is {\em exact} for noiseless data. 

The assumptions in Proposition~\ref{Old_l1} are sufficient conditions but not necessary. Our numerical simulations illustrate {\em exact recovery}
in more challenging situations, where the vicinities are not well separated (central images in the second row of Figure \ref{fig:th1}).  

For noisy data,  Proposition~\ref{Old_l1} says that it is the concept of vicinities that provide the adequate framework to look at the error between the true solution and 
the one provided by the $\ell_1$-norm minimization approach. Specifically, the error is controlled by the coherent misfit (\ref{co_}) and the incoherent remainder (\ref{in_}),
which are shown to be small   when the noise  is small in $\ell_2$. This means that the reconstructed source is supported 
mainly in the vicinities $S_j$ of the true solution, $j \in T$,  and the grass in the image is low, i.e., 
the part of the solution supported away from the vicinities $S_j$ is small.

Proposition~\ref{Old_l1} implies that a key to control the noise is the constant $\gamma$ defined in (\ref{def:gamma}). 
In general, we have $\gamma= O(\sqrt{N})$. Indeed, let $\vect y$ be the minimum $\ell_2$-norm solution of the problem $\Ac \vect \rho=\vect b$ such that its support has at most size $N$. Let $\cA_{y}$ be the submatrix  of $\cA$ that contains the columns that correspond to the non-zero entries of $\vect y$. Then, the minimum $\ell_1$ solution $\bfrho$ satisfies  (by Cauchy-Schwartz $ \| \vect x \|_{\ell_1}  \leq  \sqrt{N} \|  \vect x \|_{\ell_2} , \forall \vect x \in \mC^N$)
$$ \| \bfrho \|_{\ell_1} \le \| \vect y\|_{\ell_1} < \sqrt{N} \| \vect y\|_{\ell_2} <  \sqrt{N} \left\|   (\cA^{*}_{y} \cA_{y} )^{-1}   \cA^{*}_{y} \right\|_{\ell_2} \| \bfb  \|_{\ell_2} .$$
%since by Cauchy-Schwartz for any $\vect x \in \mC^N$
%$$ \| \vect x \|_{\ell_1}  \leq \|  \vect x \|_{\ell_2}  \sqrt{N} .$$
%
%Indeed, let  $\cA_{T}$ be the sub matrix of $\cA$ that contains the columns that correspond to the non zero entries of $\bfrho$. Then we have
%$$ \bfrho =  (\cA^{*}_{T} \cA_{T} )^{-1}   \cA^{*}_{T} \bfb $$
%and 
%%$$\gamma = \sup_{\bfb} 
%$$ \frac{ \| \bfrho \|_{\ell_1}}{\| \bfb \|_{\ell_2}} = \frac{ \|  (\cA^{*}_{T} \cA_{T} )^{-1}   \cA^{*}_{T} \bfb   \|_{\ell_1}}{\| \bfb \|_{\ell_2}} $$
%By Cauchy-Schwartz for any $\bfb \in \mC^N$
%$$ \| \bfb \|_{\ell_1}  \leq \|  \bfb \|_{\ell_2}  \sqrt{N} $$
%so that
%$$  \frac{ \| \bfrho \|_{\ell_1}}{\| \bfb \|_{\ell_2}} \ge  \sqrt{N} \frac{ \|  (\cA^{*}_{T} \cA_{T} )^{-1}   \cA^{*}_{T} \bfb   \|_{\ell_1}}{\| \bfb \|_{\ell_1}}  $$
%Therefore, 
%$$\gamma = \sup_{\bfb} \frac{ \| \bfrho \|_{\ell_1}}{\| \bfb \|_{\ell_2}}  \ge  \sqrt{N}  \sup_{\bfb}  \frac{ \|  (\cA^{*}_{T} \cA_{T} )^{-1}   \cA^{*}_{T} \bfb   \|_{\ell_1}}{\| \bfb \|_{\ell_1}} .$$
%$$ \| \bfrho \|_{\ell_2} \leq  || (\cA^{*}_{T} \cA_{T} )^{-1}   \cA^{*}_{T} \|_{\ell_2} \| \bfb \|_{\ell_2} $$
%By Cauchy-Schwartz for any $\bx \in \mC^K$
%$$ \| \bx \|_1  \leq \|  \bx \|_2  \sqrt{K} .$$
%
%
%the minimum $\ell_1$ solution of $\cA \bfrho = \bfb$ has the property that its $\ell_2$ norm satisfies, 
%$$ \| \bfrho \|_{\ell_2} \leq c \|  \bfb \|_{\ell_2} .$$
%In order to see that 
%Unfortunately, for our imaging matrices $\gamma= O(\sqrt{N})$. 
This means that
the quality of the image deteriorates as the number of measurements $N \to \infty$. The remedy that we propose to this is to augment the imaging matrix 
$\Ac$ with a ``noise collector'' $\Cc$ as described in the following Proposition.  

\begin{proposition}\label{noise_c}
 There exists a  $N \times \Sigma$ noise collector matrix $\Cc$, with $\Sigma \lesssim e^N$, such that the columns of the augmented matrix $\Dc= [\Ac \, | \,  \Cc] $
satisfy $\| \vect d_j \| =1$, 
 \begin{equation}\label{deco_3}% \label{new_o}
%  \forall  i  \hbox{ and } j \hbox{ we have }  |\langle \vect a_i,\vect c_j \rangle | < \frac{1}{3 M},
 |\langle \vect a_i,\vect c_j \rangle | < \frac{1}{3 M} \,\,\,\,  \forall  i  \hbox{ and } j \, ,
\end{equation}
\begin{equation}\label{deco_2}% \label{new_o}
%  \forall  i \neq j, \hbox{ we have }  |\langle \vect c_i,\vect c_j \rangle | < \frac{1}{3 M},
|\langle \vect c_i,\vect c_j \rangle | < \frac{1}{3 M}  \,\,\,\,  \forall  i \neq j, 
\end{equation}
% \begin{equation}\label{deco_3}% \label{new_o}
%|\langle \vect d_i,\vect d_j \rangle | < \frac{1}{3 M},\,\, \forall  i \neq j, 
%\end{equation}
 and there is a positive constant
\begin{equation}\label{gamma_est}
\gamma  \leq 18 M^2\, ,
\end{equation}
 such that
\begin{equation}\label{as:sine2}
\forall\, \vect b,\,\, \exists \,\,  \vect \rho \hbox{ such that }  \Dc \vect \rho=\vect b \hbox{ and } \| \vect \rho \|_{\ell_1} \leq  \gamma  \| \vect b \|_{\ell_2}. 
\end{equation}
\end{proposition}
 \begin{proof} Let $\vect d_{i}=\vect a_{i}$, for $i=1,\dots,K$.
 We will construct iteratively a sequence of vectors $\vect d_{K+1}=\vect c_{1}$, $\vect d_{K+2}=\vect c_{2}$, $\dots$, $\vect d_{K+\Sigma}=\vect c_{\Sigma}$
 such that for each $s =1 \dots \Sigma$ 
 \[
 | \langle \vect d_{k}, \vect d_{K+s} \rangle | \leq \frac{1}{3 M},\, \forall k < s+K.
\] The iteration will terminate at a finite step, say, $\Sigma$. At the termination step we will have that for any $\vect b$, $\| \vect b \|_{\ell_2}=1$ there exists $k \leq \Sigma+K$ such that
 \begin{equation}\label{est_gamma}
 | \langle \vect d_{k}, \vect b \rangle | >  \frac{1}{3 M}.
 \end{equation}
 The finite time termination is a consequence of a volume growth estimate. Namely, if~(\ref{deco_2}) holds for all $i \neq j \leq \Sigma$, then the points $\vect c_{i}$, $i=1,2,\dots \Sigma$ are centers of non-overlapping 
 balls of radius $r$. The radius is bounded below:
 \[
 r  > \frac{1}{2} \alpha, \hbox{ where } \alpha=  \sqrt{1 - \frac{1}{9 M^2}}.
 \] 
 Thus the iteration will terminate at a finite step.
Furthermore,  if $r< \sqrt{2}$ then the number $\Sigma \lesssim e^{N \log{\frac{\sqrt{2}}{r}}}$ as the dimension $N \to \infty$, because $(r/\sqrt{2})^N \Sigma \sim 1$. 

 Let us finally estimate $\gamma$ in~(\ref{gamma_est}). Without loss of generality, we may assume $\| \vect b \|_{\ell_2}=1$. By our construction, there exists $k \leq \Sigma+K$ such that~(\ref{est_gamma})
holds. Thus we can choose $\vect d_{n_1}$ and $c_1$ so that $|c_1| \leq 1$ and $\vect b_1= \vect b - c_1  \vect d_{n_1}$
satisfies $\| \vect b_1 \|_{\ell_2} \leq  \alpha$. 
Using~(\ref{est_gamma}) inductively we can find a sequence  $\{ \vect d_{n_i}\}_{i=1}^{\infty}$,  and a sequence  $\{ c_i \}_{i=1}^\infty$, so that $|c_i|\leq \alpha^{i-1}$ and the vectors  
$\vect b_n= \vect b - \sum_{i=1}^n c_i   \vect d_{n_i}$ satisfy
$\| \vect b_n \|_{\ell_2} \leq \alpha^{n}$.
Therefore,
\begin{equation}
 \label{gamma1}
 \vect b = \sum_{i=1}^\infty c_i  \vect d_{n_i}
\end{equation}
and 
\begin{equation}
 \label{gamma2}
\| \vect \rho \|_{\ell_1}  \leq \sum_{i=1}^\infty |c_i|\leq \sum_{i=1}^\infty \alpha^{i-1}=\frac{1}{1-\alpha} \leq 18 M^2
\end{equation}
by the triangle inequality.
\end{proof}

Proposition~\ref{noise_c} is an important result as it shows that the constant $\gamma$ in (\ref{def:gamma}) 
can be made independent of $N$ by augmenting the columns of the linear system with  columns of a noise collector matrix $\Cc$. %proved in \ref{dos}. 
The columns of $\Cc$ are required to be 
decoherent to the columns of $\cA$ (see (\ref{deco_3})), and decoherent between them (see (\ref{deco_2})). 
Recalling that the columns of $\cA$ for the imaging problem are Green's vectors corresponding to points in the imaging 
window, we stress that the columns of $\cC$ do not admit a physical interpretation. They do not correspond to any 
points in the imaging window or elsewhere. Similarly, the $\Sigma$ last components of the augmented unknown vector $\bfrho$ 
in  (\ref{as:sine2}) do not have a physical meaning. They correspond to fictitious auxiliary unknowns that are introduced 
to regularize the $\ell_1$-norm minimization problem.% so that  $\gamma$ is now independent of the number of 
%measurements $N$. The key being that the noise vector $\delta \bfb=\bfb_\delta - \bfb$ is well approximated by a 
%linear combination of the columns of $\cC$. 

%\textcolor{blue}{We remark that  $\Rc(\cA)$, the range of $\cA$,  is in general small compared to the $N$-dimensional sphere in our imaging problems. 
%The $\ell_1$-norm minimization provides exact recovery in the noiseless case for which $\bfb \in \Rc({\cA})$. 
%When noise is added to the data, $\bfb \notin \Rc({\cA})$. 
%The noise collector is constructed so as to increase the range of $\Dc$ so that $\bfb$ can be approximated 
%as linear combination of the columns of $\Dc$ with the $\ell_1$-norm of the coefficients bounded by a constant independent of $N$; see (\ref{gamma1})-(\ref{gamma2}).  
%}

The drawback in this theory is that the size of the noise collector is exponential $\Sigma \lesssim e^N$. This makes it impractical. Our numerical experiments, however, indicate 
great improvement in the performance of $\ell_1$-norm minization with $\Sigma \lesssim 10K$ when the columns of $\cC$ are selected at random 
(its entries are i.i.d. Gaussian random variables with mean zero and variance $1/N$).  This works well for additive mean zero uncorrelated noise. 
For other types of noise, the idea is to construct a library that represents the values that the noise vector $\delta \bfb$ takes. It is the elements of this library 
that should be used as columns of the noise collector matrix $\cC$.  
A different approach can be followed when the noise $\delta \vect b$ is sparse so its $\ell_1$-norm is small. Then, $\Cc$ could be simply taken as the 
 $N \times N$ identity matrix $I$. % and then~(\ref{as:sine2}) could be replaced by 
% \begin{equation}\label{as:sine3}
%\forall \vect b\, \exists \vect \rho \hbox{ such that }  \Dc \vect \rho=\vect b \hbox{ and } \| \vect \rho \|_{\ell_1} \leq  \gamma  \| \vect b \|_{\ell_1}. 
%\end{equation}
This approach has been proposed and analyzed in \cite{Baraniuk09} and provides exact recovery for sparse noise vectors $\delta \bfb$. 

In the next section we present numerical results to illustrate the relevance of our theory in imaging sparse sources. 
We focus our attention in the case of additive mean zero uncorrelated noise which is not sparse. The results show  a dramatic 
improvement  using the noise collector.

\section{Imaging results in the framework of propositions \ref{Old_l1} and \ref{noise_c}}
\label{sec:vicinities}
We illustrate here the relevance of Propositions \ref{Old_l1} and \ref{noise_c} in imaging. 
We compare the solution $\vect \rho_{\ell_1}$ obtained with the $\ell_1$-norm minimization algorithm GeLMA \cite{Moscoso12},
and the $\ell_2$-norm Kirchhoff migration solution (\ref{sol_ell2}).
%\begin{equation} \label{sol_ell2}
%\vect \rho_{\ell_2} = \Ac^* \vect \bfb,
%\end{equation}
%where $\Ac^*$ is the conjugate transpose of $\Ac$.  
Our results illustrate:  
\begin{enumerate}
\item The well-known super-resolution for $\ell_1$,   
meaning that $\bfrho_{\ell_1}$ determines the support of the unknown $\bfrho$ with higher accuracy than the conventional resolution limits, provided the assumptions of Proposition \ref{Old_l1}
 are satisfied.
\item The equally well known sensitivity of $\ell_1$ to additive noise. This is made more precise in the imaging context where the constant $\gamma$ in (\ref{def:gamma}) grows with the number of measurements as $\sqrt{NS}$, where $NS$ is the total number of measurements acquired by $N$ receivers at $S$ frequencies. We observe that, for a given level of noise, the $\ell_1$-norm reconstruction deteriorates as the number of measurements increases.  
\item  The noise collector matrix $\cC$ stabilizes $\ell_1$-norm minimization in the presence of noise. 
\end{enumerate}
We also show how the bandwidth, the array size, and the number of sources affect the vicinities defined in (\ref{vicinity}). The numerical results are not specialized to a particular physical regime. They illustrate only the 
role of the Propositions \ref{Old_l1} and \ref{noise_c} in solving the associated linear systems. 

\subsection*{Imaging setup}

The images are obtained in a homogeneous medium with an active array of $N=25$ transducers. We collect measurements corresponding to $S=25$ frequencies equispaced in the bandwidth. Thus, the length of the data vector $\bfb$ is $N S=625$.  The ratios between the array size $a$ and the distance $L$ to IW, and between the bandwidth $2B$ and the central frequency $\omega_0$ vary in the numerical experiments, so the classical Rayleigh resolution limits change. The size of the IW is fixed. It is discretized using a uniform grid of $K=3721$ points of size $\lambda_0/2$  in range and cross-range directions.

The images have been formed by solving the $\ell_1$-norm minimization problem (\ref{l1normsol}) using the algorithm GeLMA in \cite{Moscoso12}. GeLMA is an iterative shrinkage-thresholding algorithm
that provides the exact solution, for noiseless data, independently of the value of the parameter $\tau$ (see (\ref{min-max})) used to promote the sparsity of the images.

\subsection*{Results for noiseless data. Super-resolution and $\ell_1$-reconstructions}
%%%%%%%%%%%%%%%%%%%
\begin{figure}[h]
\begin{center}
\begin{tabular}{cccc}
\includegraphics[scale=0.18]{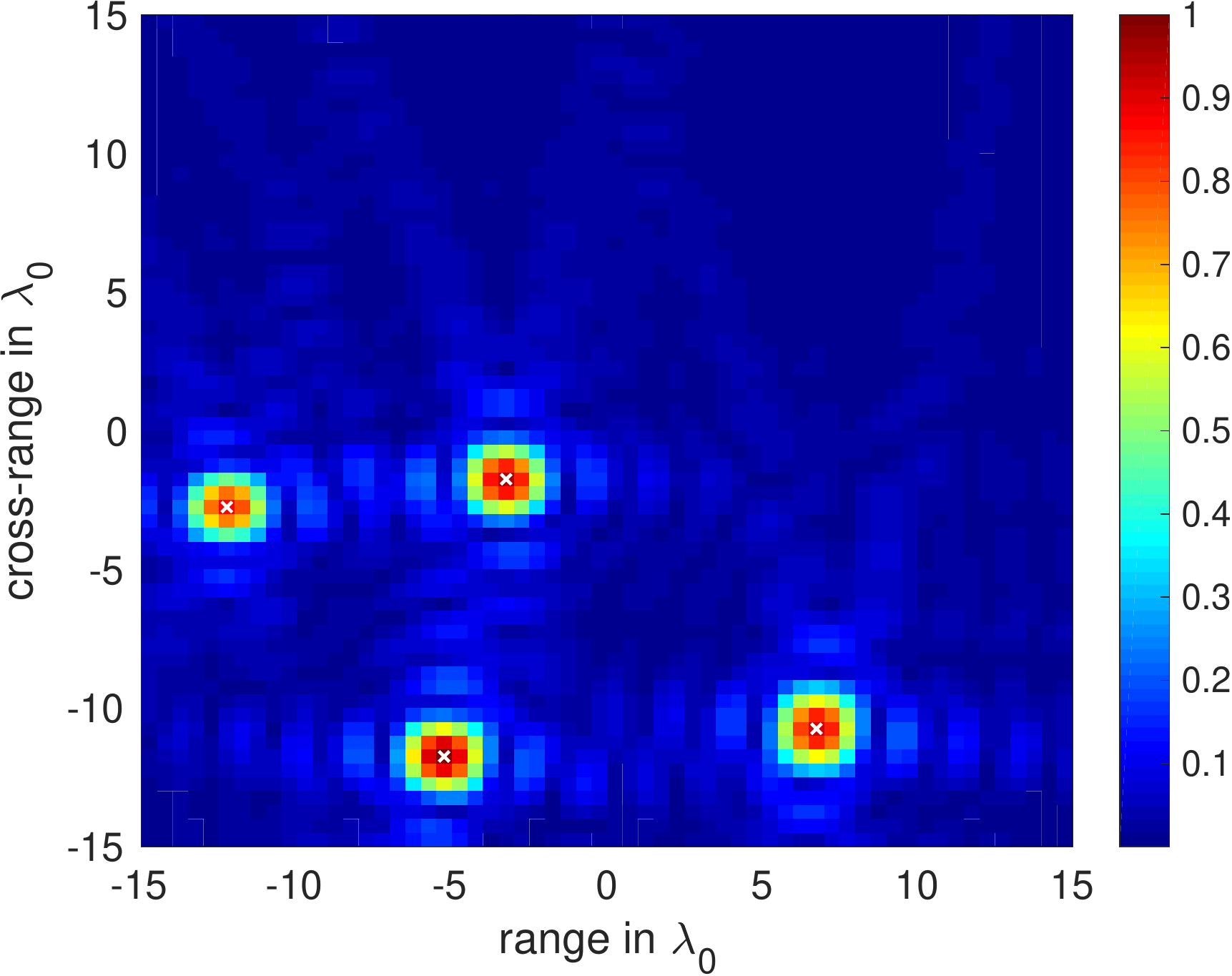}&
\includegraphics[scale=0.18]{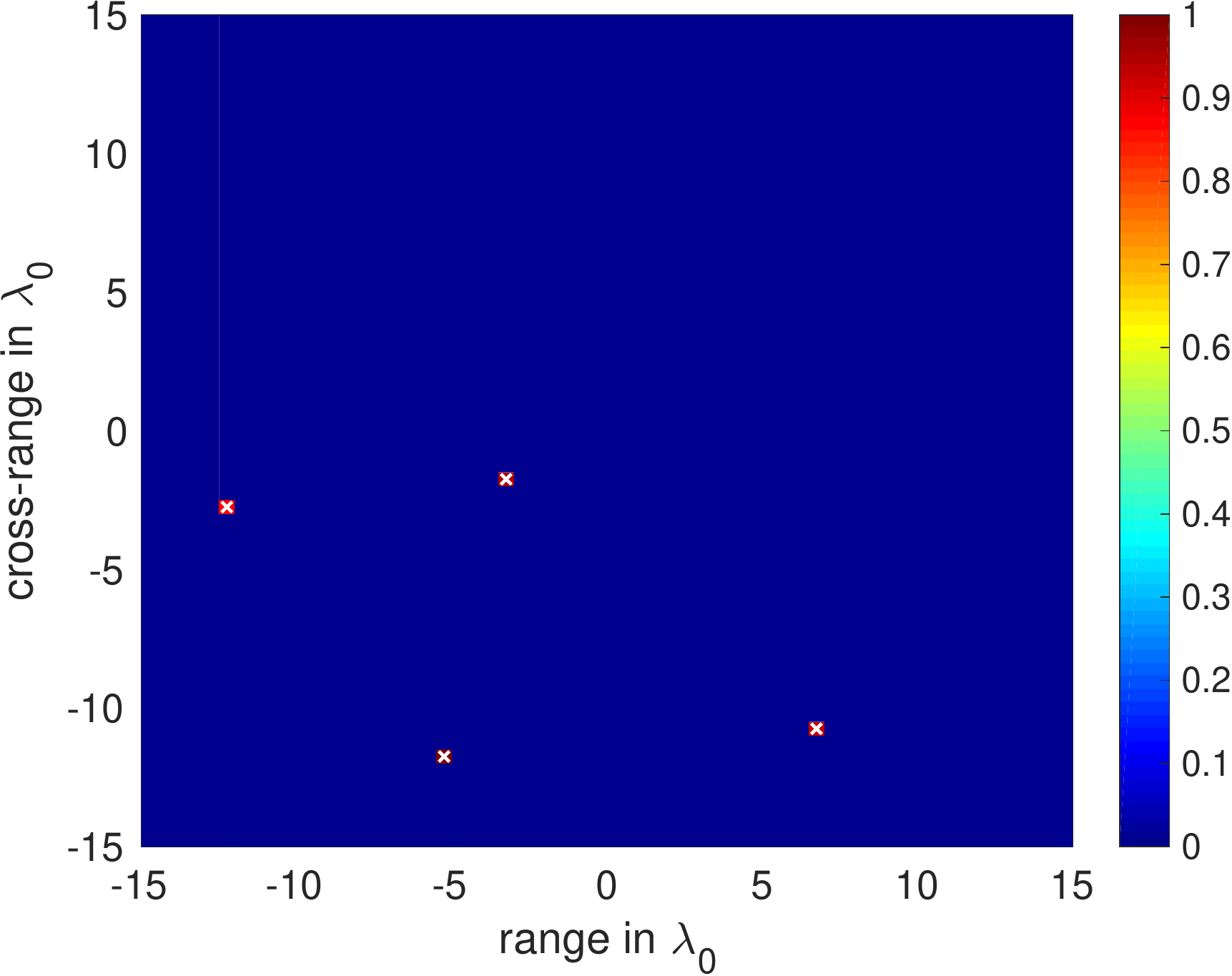}&
\includegraphics[scale=0.18]{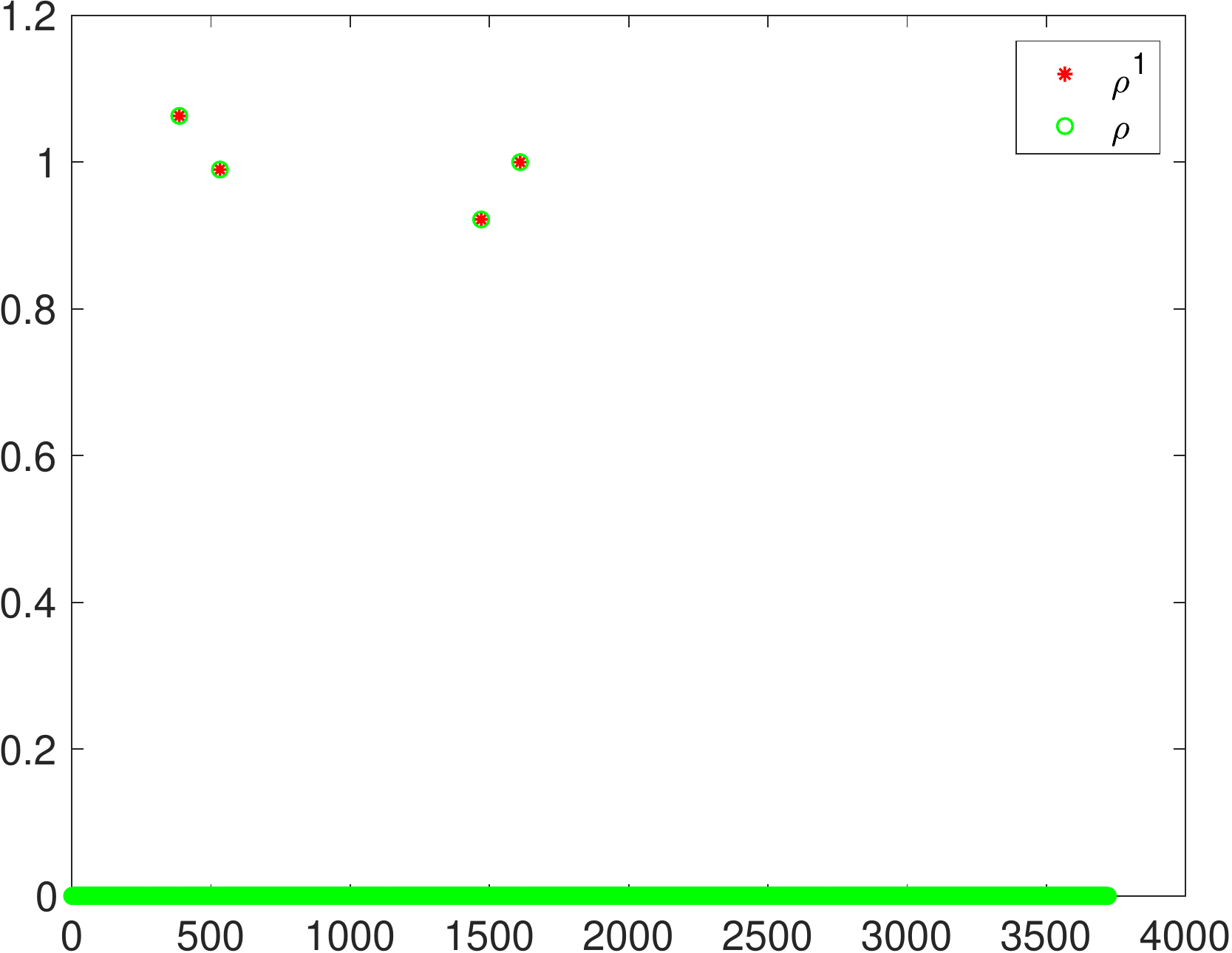} &
\includegraphics[scale=0.18]{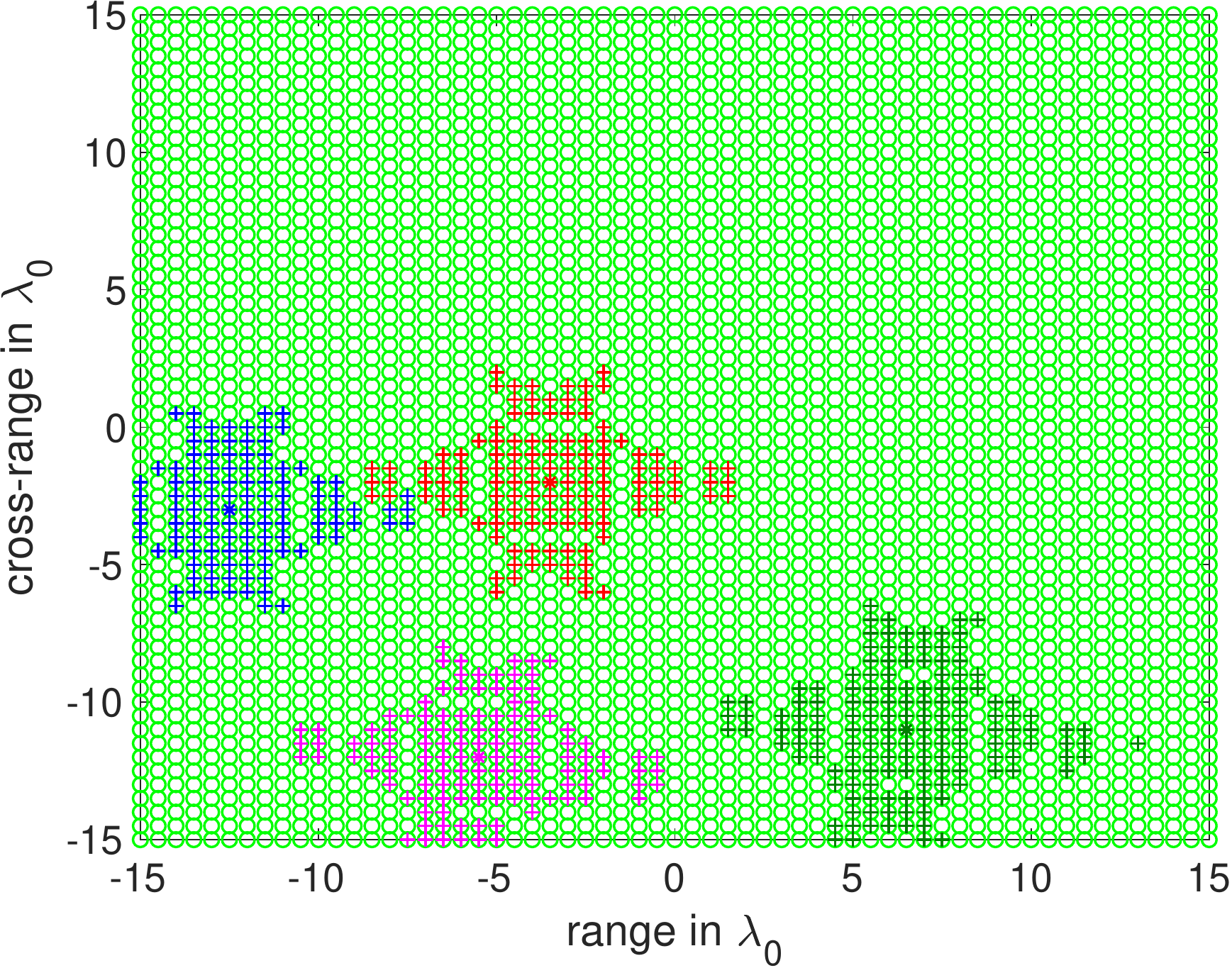}\\
\end{tabular}
\begin{tabular}{cccc}
\includegraphics[scale=0.18]{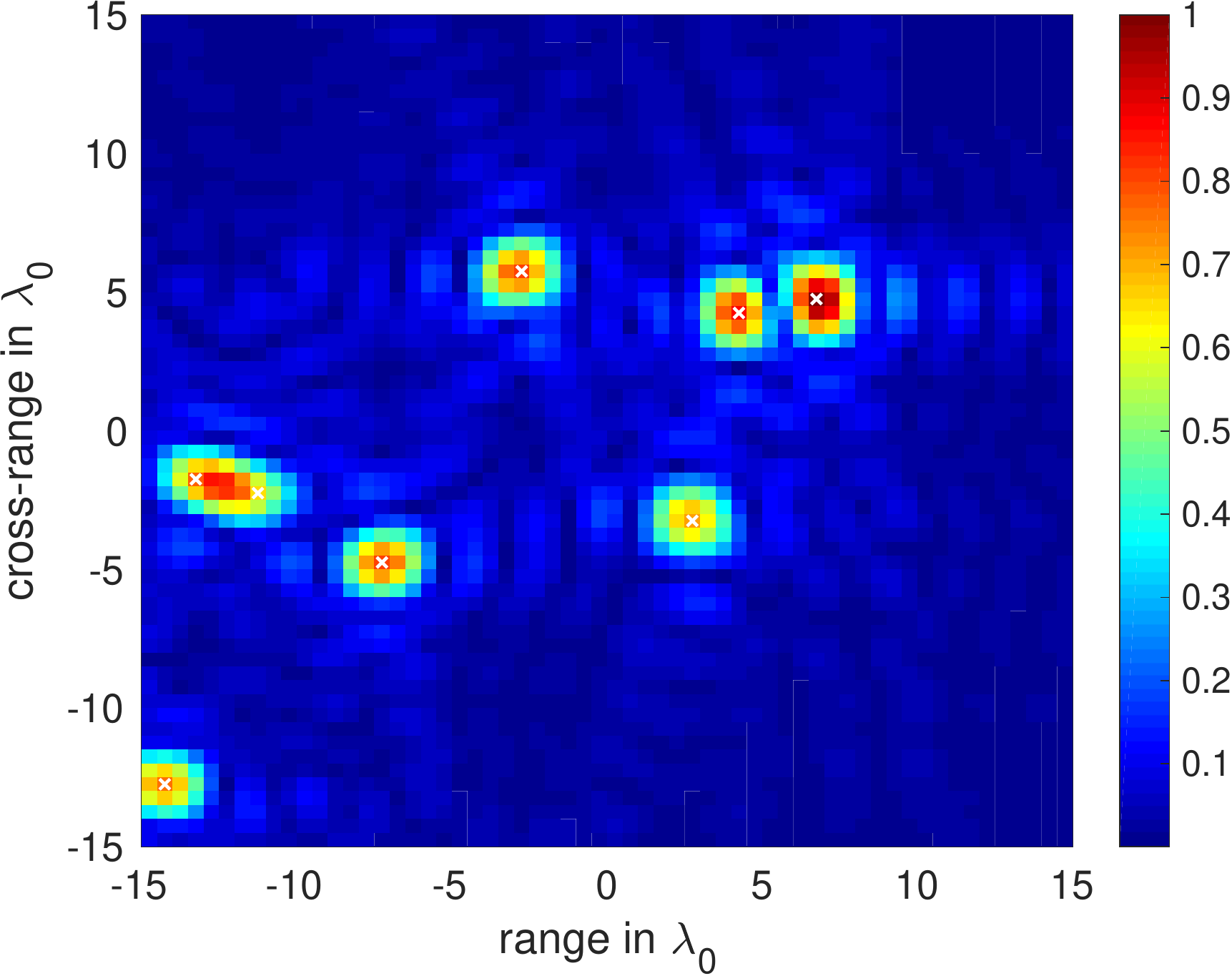}&
\includegraphics[scale=0.18]{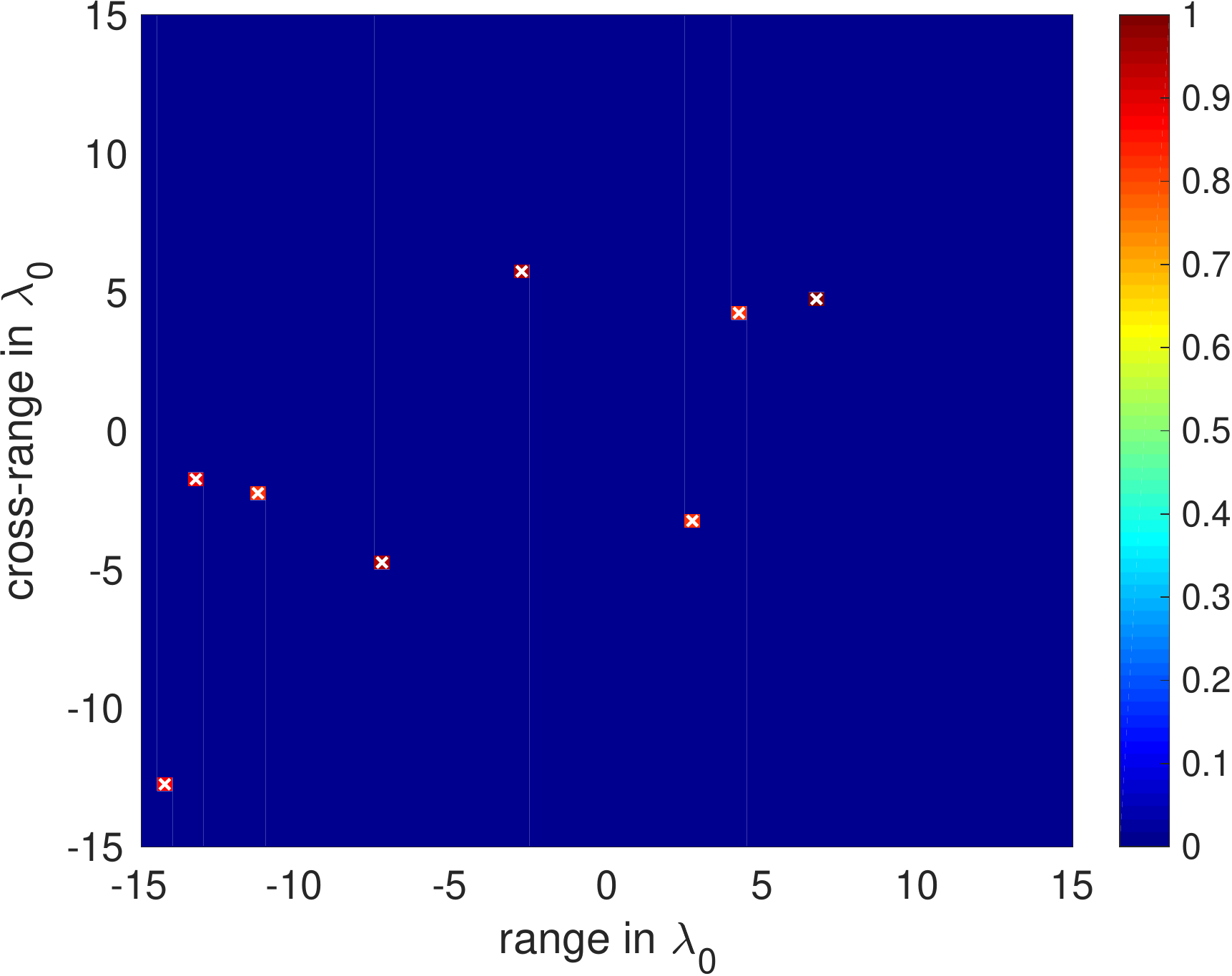}&
\includegraphics[scale=0.18]{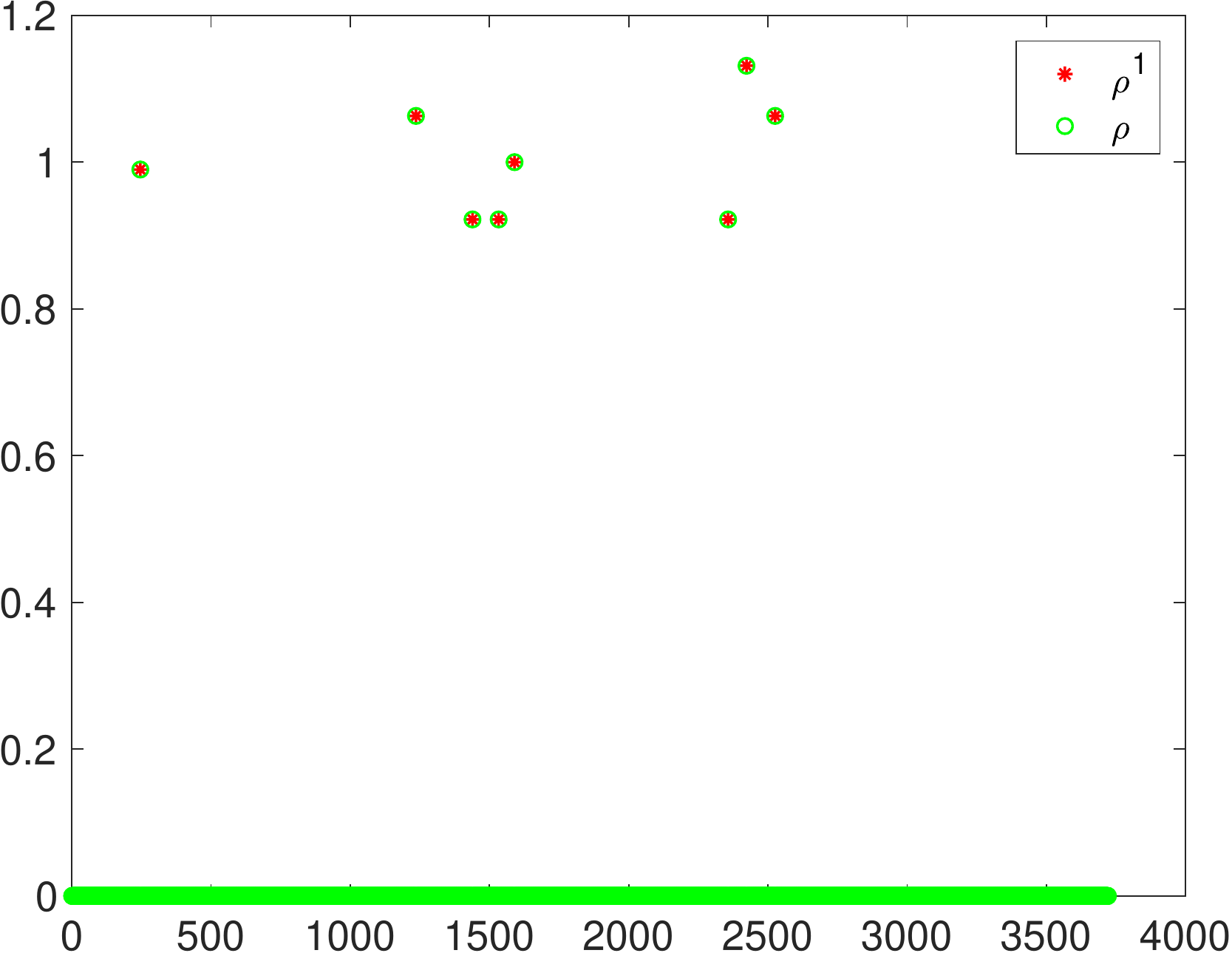} &
\includegraphics[scale=0.18]{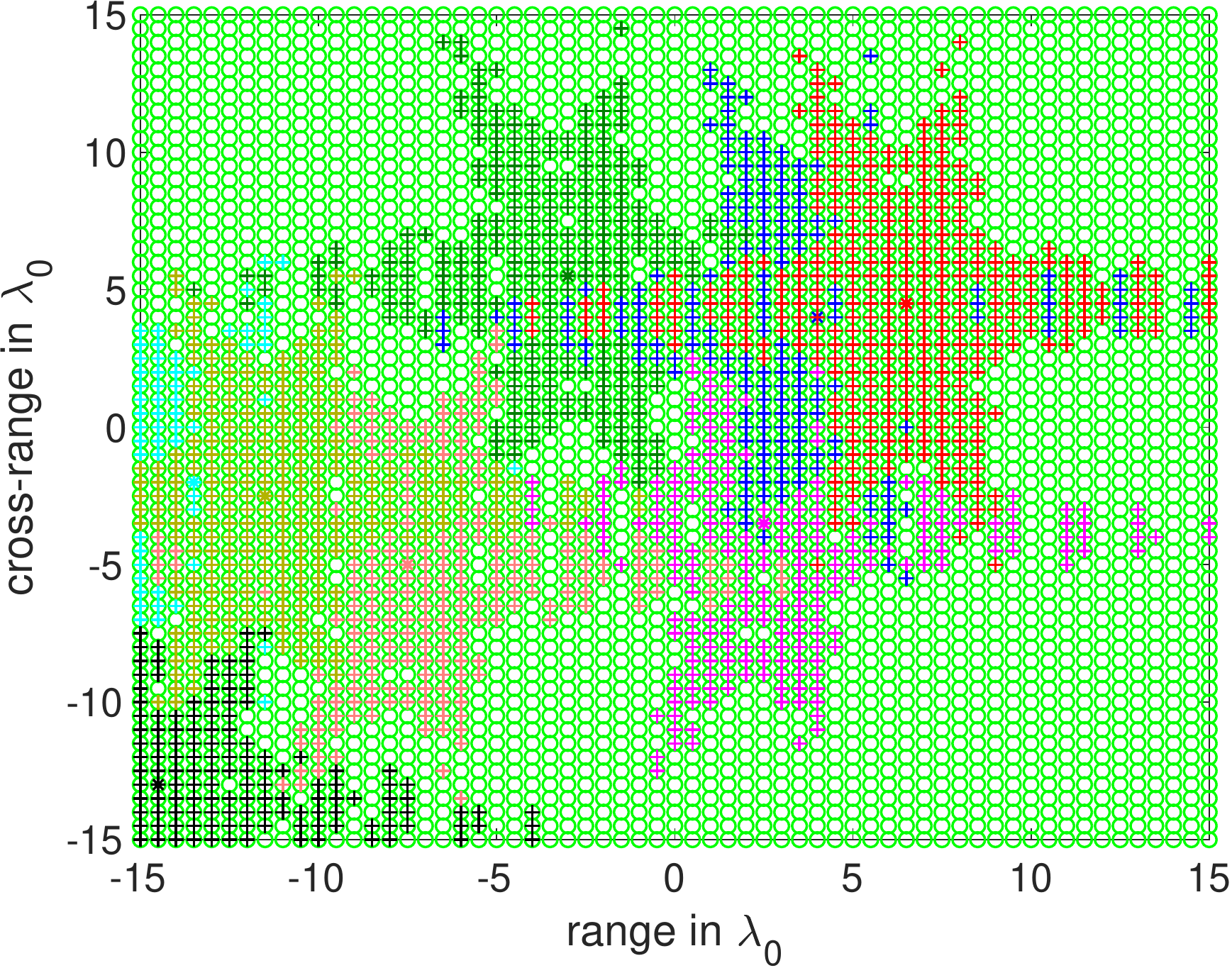} \\
\end{tabular}
\end{center}
\caption{Imaging with no noise $M=4$ (top row) and $M=8$ (bottom row) targets. From left to right: $\vect \rho_{\ell_2} $,
$\vect \rho_{\ell_1}$, comparison between $\vect \rho_{\ell_1}$ (red stars) and the true solution (green circles), 
and the vicinities $S_j$, $j=1,\ldots,M$, plotted with different colors. %Top row $M=4$, bottom row $M=8$. 
Large array aperture and large bandwidth; $a/L=1/2$ and $(2B)/\omega_0=1/2$. 
}
\label{fig:th1}
\end{figure}
%%%%%%%%%%%%%%%%%%%

Figure \ref{fig:th1} shows the results obtained for a relatively large array and a relatively large bandwidth corresponding to ratios $a/L=1/2$ and $(2B)/\omega_0=1/2$ when the data is noiseless. 
From left to right we show the $\vect \rho_{\ell_2}$ solution (\ref{sol_ell2}), the $\vect \rho_{\ell_1}$ solution obtained from (\ref{l1normsol}),  the comparison between $\vect \rho_{\ell_1}$  (red stars) and the true solution $\vect \rho$ (green circles),
and the vicinities $S_j$  defined in (\ref{vicinity}) plotted with different colors.  The top and bottom rows show images with
$M=4$ and $M=8$ sources, respectively. The exact locations of the sources are indicated with white crosses in the two leftmost columns. 
The $M=4$ sources in the top row are very far apart: their vicinities do not overlap as it can be seen in the top right image. 
In this case, all the conditions of Proposition \ref{Old_l1} are satisfied and we find the 
exact source distribution  by  $\ell_1$-norm minimization.
The $M=8$ sources in the bottom row are closer, and their vicinities are larger; according to (\ref{vicinity}) the size of the vicinities increases with $M$.
In fact, their vicinities overlap as it can be seen in the bottom right image. Still, the $\ell_1$-norm minimization algorithm finds the exact solution. % $\vect \rho$.

The classical resolution limits for this setup are $c_0/(2B)=2\lambda_0$ in range and $\lambda_0 L/a=2\lambda_0$ in cross-range. This means that the resolution of the $\ell_2$-norm solutions is of the order $O(2\lambda_0)$; see the left column of Figure \ref{fig:th1}. Recall that our discretization is $\lambda_0/2$, that is four times finer than the classical resolution limit. Thus, each source roughly corresponds to a four by four pixel square, which is what the $\vect \rho_{\ell_2}$ solutions show. Note that for $M=8$, because two sources are quite close, the $\vect \rho_{\ell_2}$ solution only displays $7$ sources.  The ability of $\ell_1$-norm minimization to determine the location of the sources with a better accuracy than the classical resolution limits is referred to as {\em super-resolution}. 
%Super-resolution can be  observed only if the discretization of the IW is fine enough so that vicinities exist. 

We stress that if the IW is discretized using a very fine grid, with grid size smaller than the classical resolution limit, then the columns of the matrix $\Ac$ are almost parallel and the decoherence condition (\ref{ortho}) is violated. The columns that are almost parallel to those indexed by the support of the true solution are contained in the vicinities (\ref{vicinity}). The number of columns that belong to the vicinities depend on the imaging system.
To illustrate the effect of the array and bandwidth sizes on the size of the vicinities we plot in Figure \ref{fig:thvic} the vicinity of one source for  $M=4$. From left to right we plot the vicinities for $[a/L,2B/\omega_0]=[1/2,1/2]$, $[a/L,2B/\omega_0]=[1/2,1/4]$, $[a/L,2B/\omega_0]=[1/4,1/2]$, and $[a/L,2B/\omega_0]=[1/4,1/4]$.
%$a/L=1/2$ and $(2B)/\omega_0=1/2$, $a/L=1/2$ and $(2B)/\omega_0=1/4$, $a/L=1/4$ and $(2B)/\omega_0=1/2$ and and $a/L=1/4$ and $(2B)/\omega_0=1/4$. 
As expected, the size of the vicinity is proportional to the resolution estimates $\lambda_0 L/a$ and $c_0/(2B)$ in  cross-range and range, respectively. 

%%%%%%%%%%%%%%%%%%% 
\begin{figure}[t]
\begin{center}
\begin{tabular}{cccc}
$\scriptstyle a/L=1/2,\  (2B)/\omega_0=1/2$ & $\scriptstyle a/L=1/4,\  (2B)/\omega_0=1/2$ & $\scriptstyle a/L=1/2,\  (2B)/\omega_0=1/4$ &  $\scriptstyle a/L=1/4,\  (2B)/\omega_0=1/4$ \\
\includegraphics[scale=0.18]{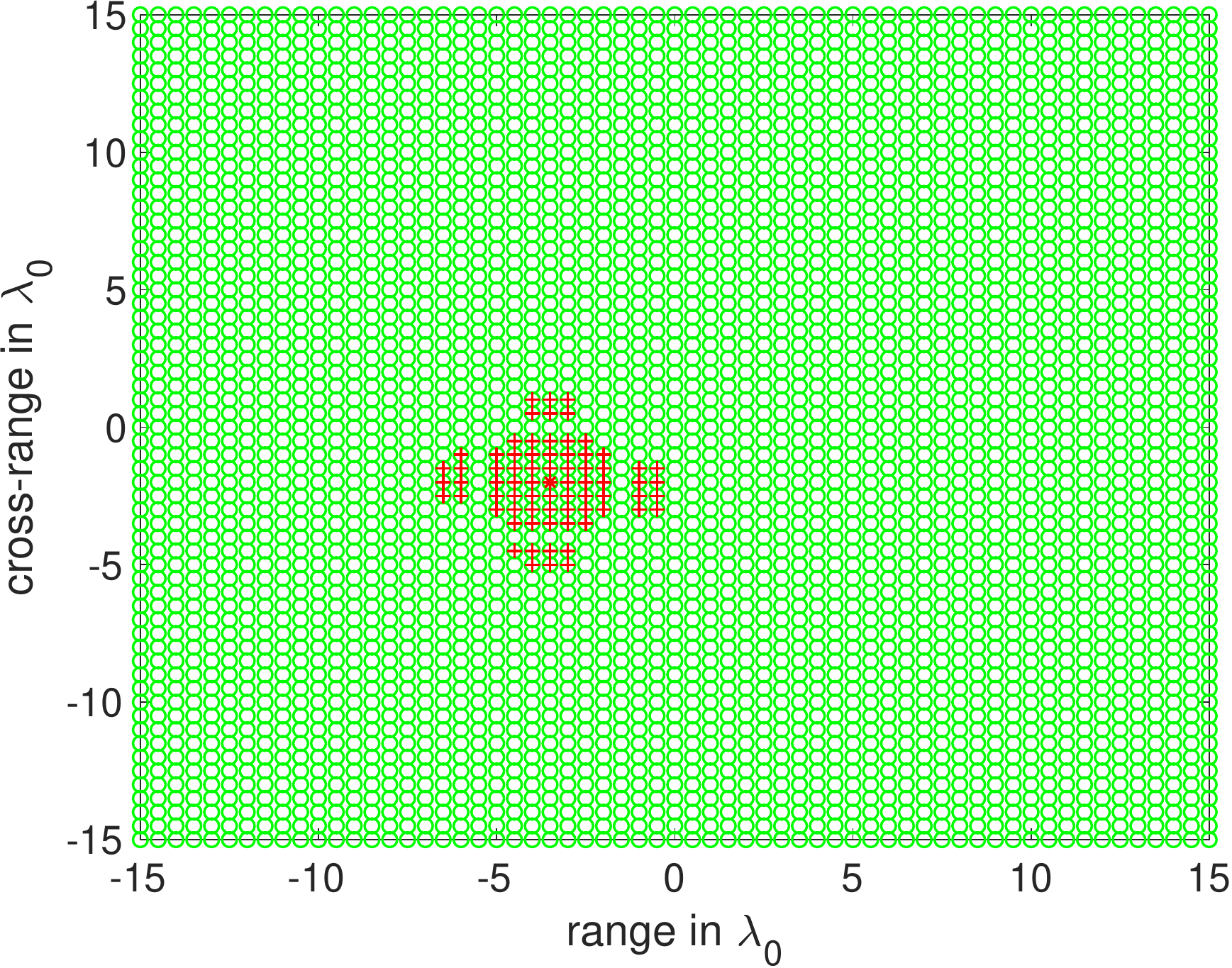}& 
\includegraphics[scale=0.18]{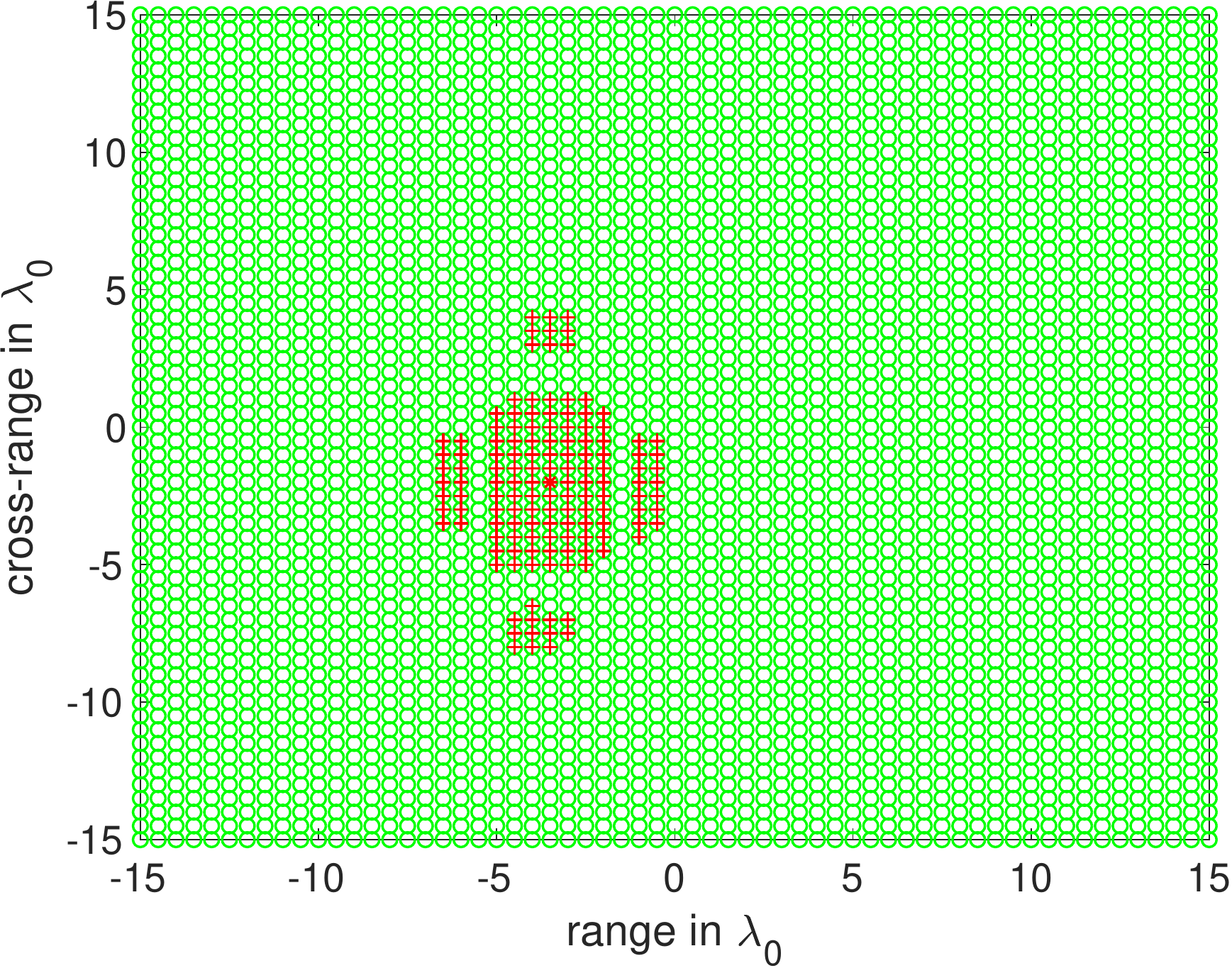}&
\includegraphics[scale=0.18]{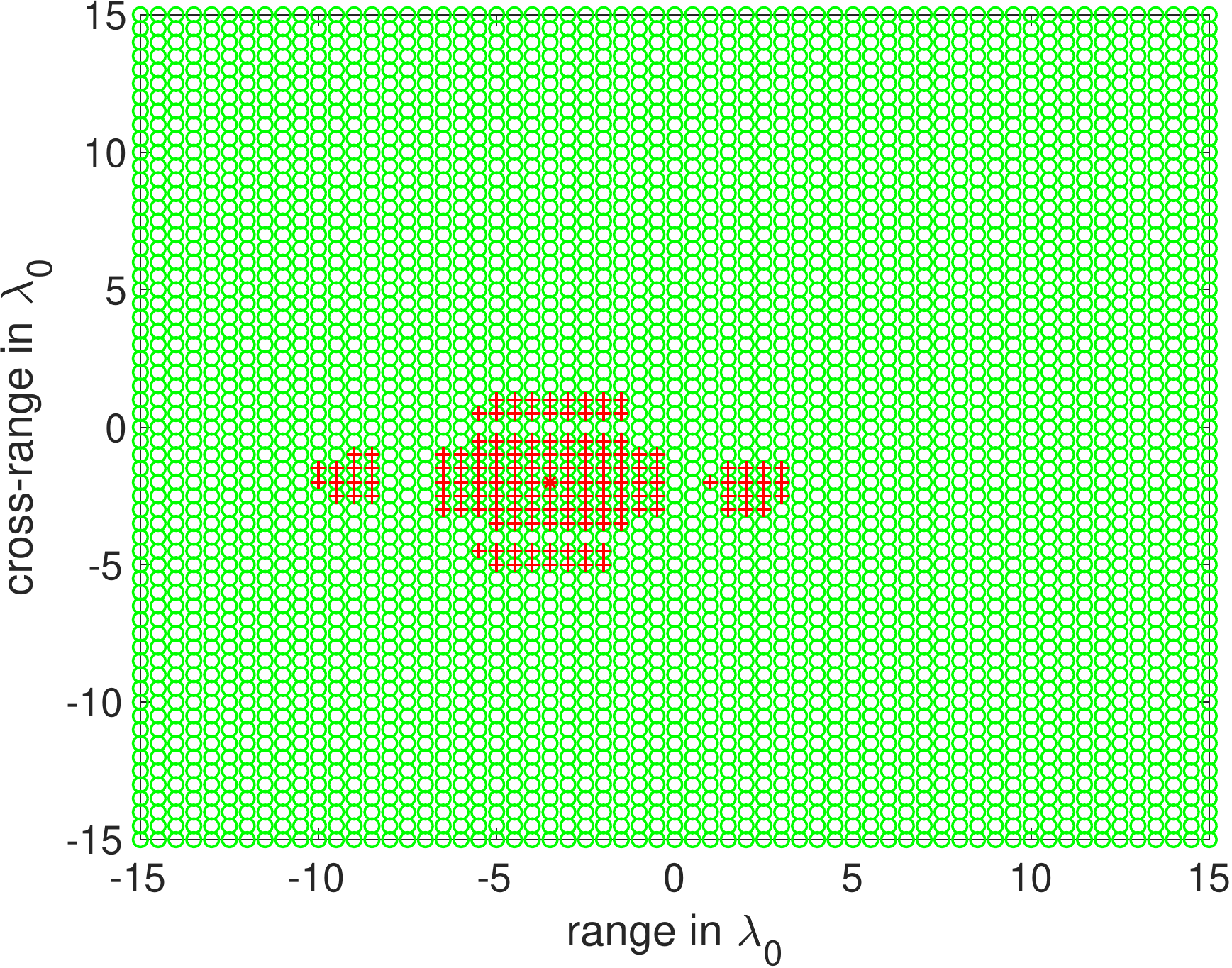} &
\includegraphics[scale=0.18]{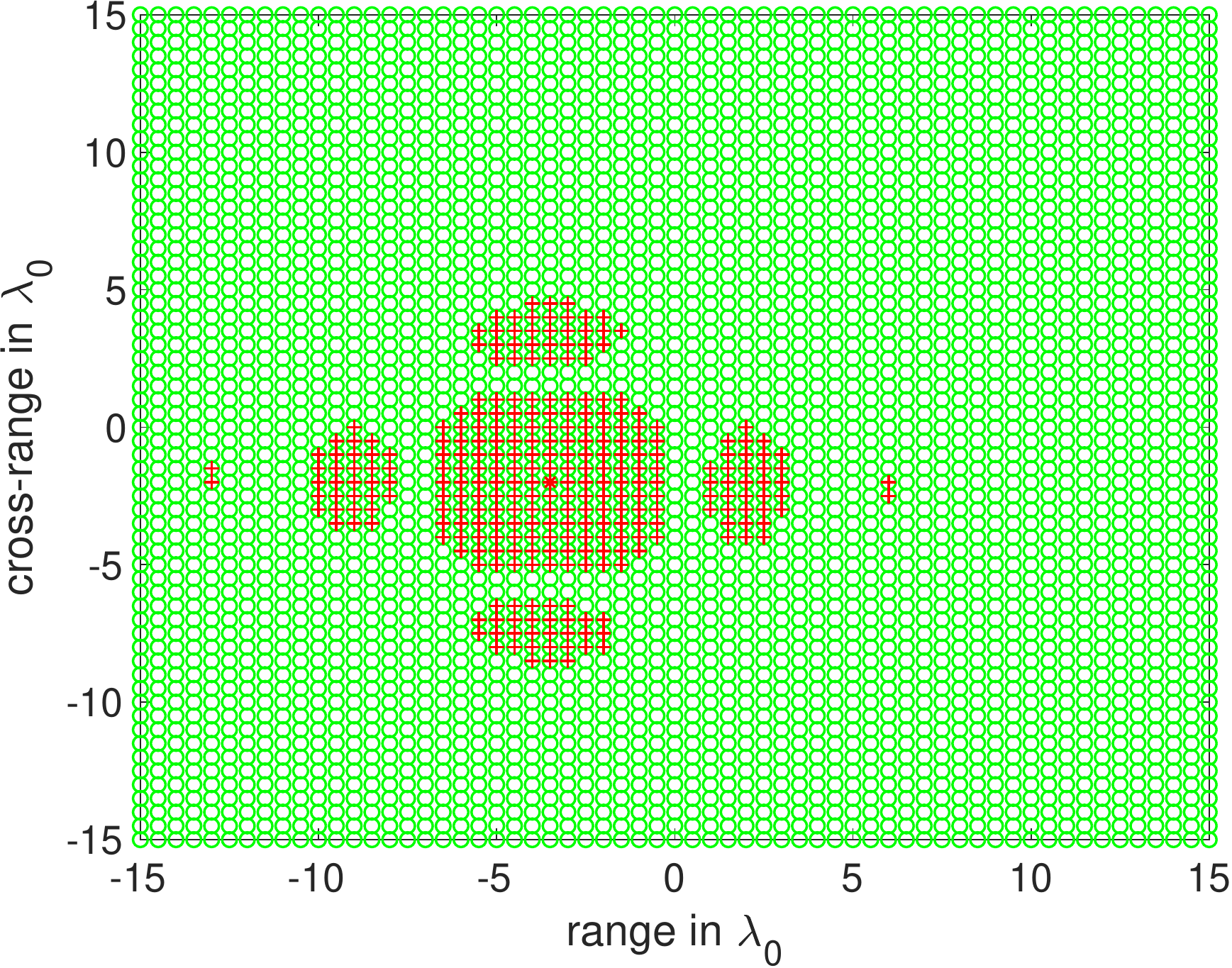}\\
\end{tabular}
\end{center}
\caption{Vicinities (\ref{vicinity}) for different array and bandwidth sizes. From left to right the ratios $(a/L,2B/\omega_0)$ are: 
$(1/2,1/2)$, $(1/2,1/4)$, $(1/4,1/2)$ and $(1/4,1/4)$. }
\label{fig:thvic}
\end{figure}
%%%%%%%%%%%%%%%%%%% 

\subsection*{Results for noisy data. Stabilization of $\ell_1$-norm minimization using the noise collector matrix $\Cc$}
We add now mean zero uncorrelated noise to the data. We examine the results for different values of the signal to noise ratio (SNR). 
We consider first the same imaging configuration as in Figure \ref{fig:th1} with $M=4$ sources. 
The number of data is $NS=625$ and the number of unknowns is $K=3721$. 
In the top row of Figure \ref{fig:noisec1} we plot the minimal $\ell_1$-norm  image obtained by solving problem (\ref{l1normsol}) when the SNR is  $4$dB.
The true solution is shown with white crosses. It is apparent that, even for  this moderate level of noise,  $\ell_1$-norm minimization  %(\ref{l1normsol} 
fails to give a good image. 

The problem can be alleviated using the noise collector matrix $\Cc$, as it can be seen in the results shown in the bottom row of Figure \ref{fig:noisec1}. To construct the noise collector matrix $\Cc$ that verifies the assumptions of Proposition \ref{noise_c}, we take its columns $\vect c_j$ to be random vectors in $\mathbb{C}^{NS}$ with mean zero and variance $1/(NS)$. Their $\ell_2$-norm tends to one as $NS\rightarrow\infty$, and we check that conditions (\ref{deco_3}) and (\ref{deco_2}) are satisfied. 
In theory, the number of columns $\Sigma$ should be very large, of the order of $e^{NS}$, but in practice,  we obtain stable results with $\Sigma$ of the order of $10^4$, which is roughly $3K$.

The solution $\vect \bfrho_{\ell_1} \in\mathbb{C}^{K+\Sigma}$ obtained with the noise collector can be decomposed into two vectors; the vector $\vect \bfrho_{iw}\in\mathbb{C}^{K}$ 
corresponding to the sought solution in the $IW$, and the vector $\vect \bfrho_{noise}\in\mathbb{C}^{\Sigma}$ 
that absorbs the noise. We display these two vectors in the bottom right plot of Figure \ref{fig:noisec1}. The first $K$ components correspond to $\vect \bfrho_{iw}$ and the remaining $\Sigma$ components
to $\vect \bfrho_{noise}$. It is remarkable that  the vector  $\vect \bfrho_{iw}$ is very close to the true solution and that it contains only some small {\em grass}. This means that both the coherent misfit (\ref{est_co-}) and the incoherent remainder (\ref{est_in-}) are now small. This is in accordance with the theoretical error estimates (\ref{est_co-}) and (\ref{est_in-}), where $\gamma$ is now independent of the dimension of the data vector $NS$; see (\ref{gamma_est}).%\textcolor{red}{ Note:  $\sqrt{N S}=25$ while $18M^2=288$!}

\begin{figure}[htbp]
\begin{center}
\begin{tabular}{cc}
\includegraphics[scale=0.28]{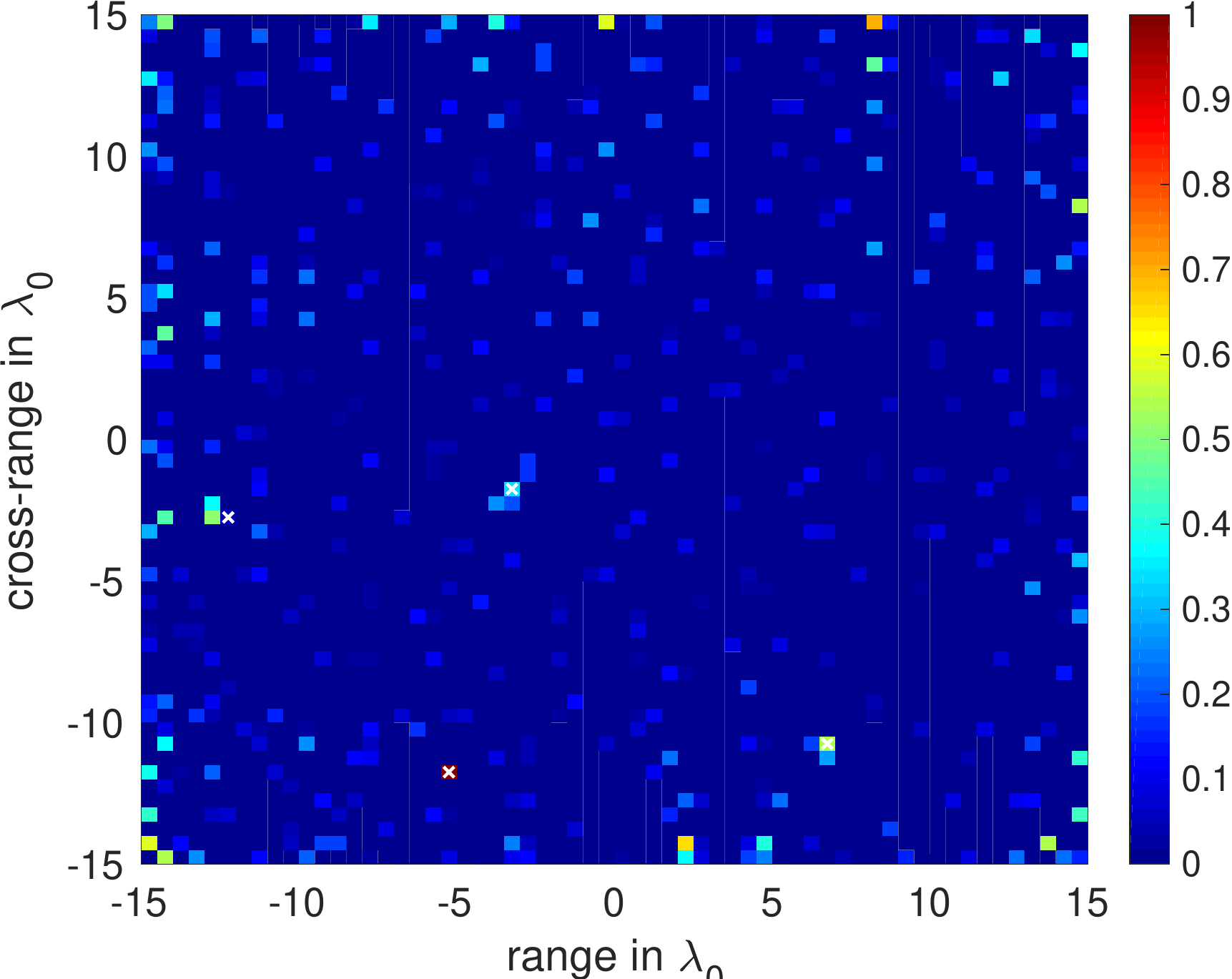}&
\includegraphics[scale=0.28]{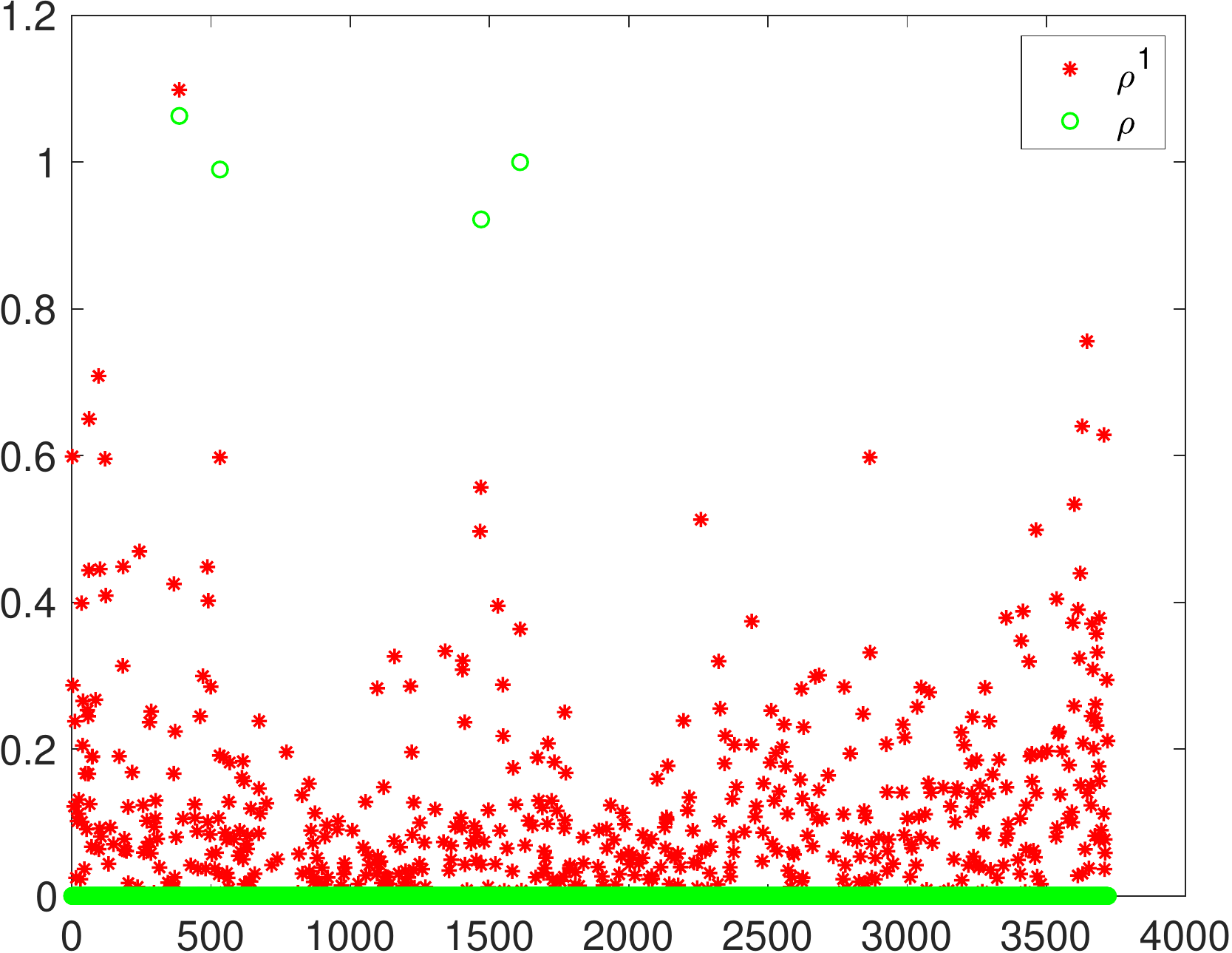}\\
\includegraphics[scale=0.28]{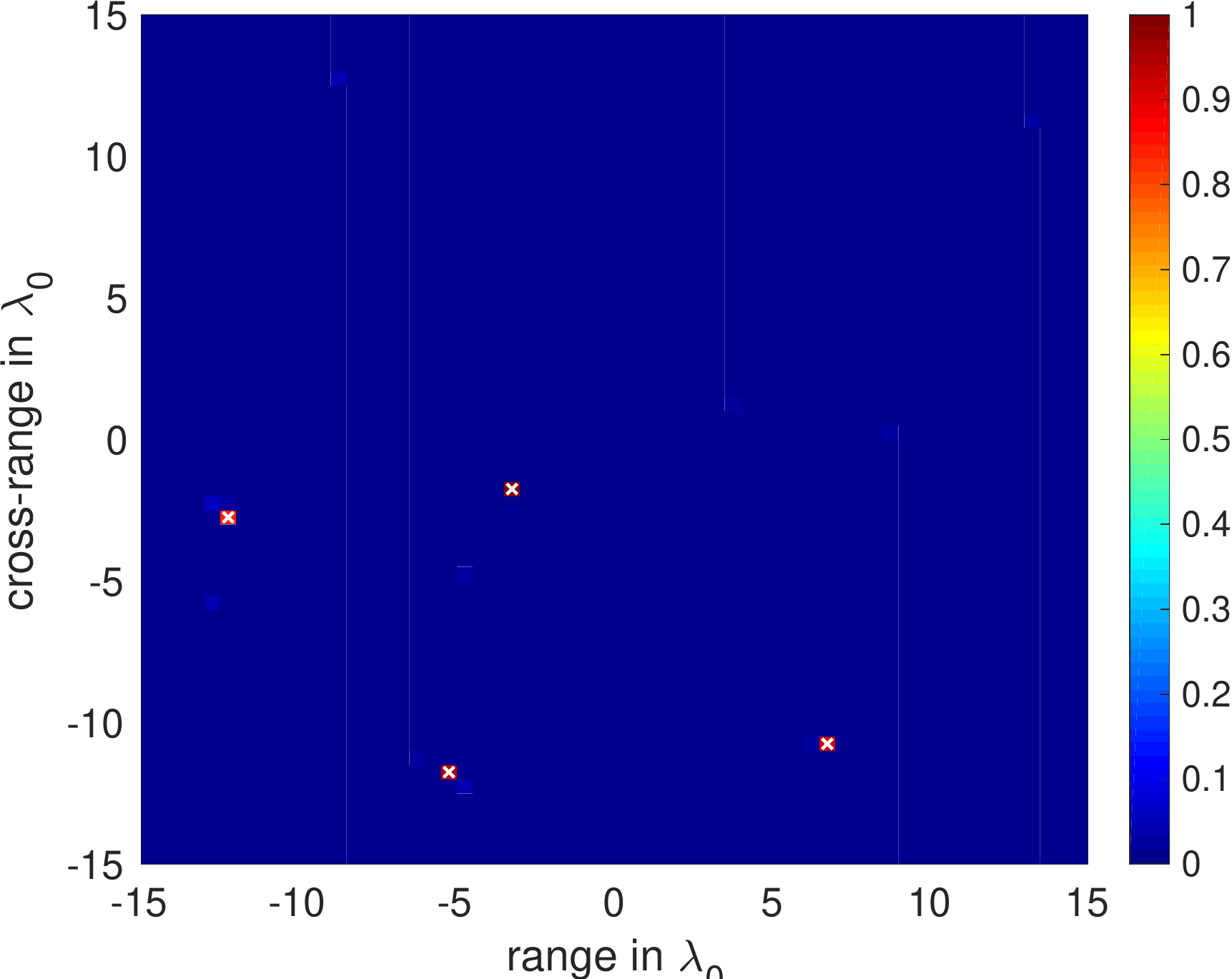}&
\includegraphics[scale=0.28]{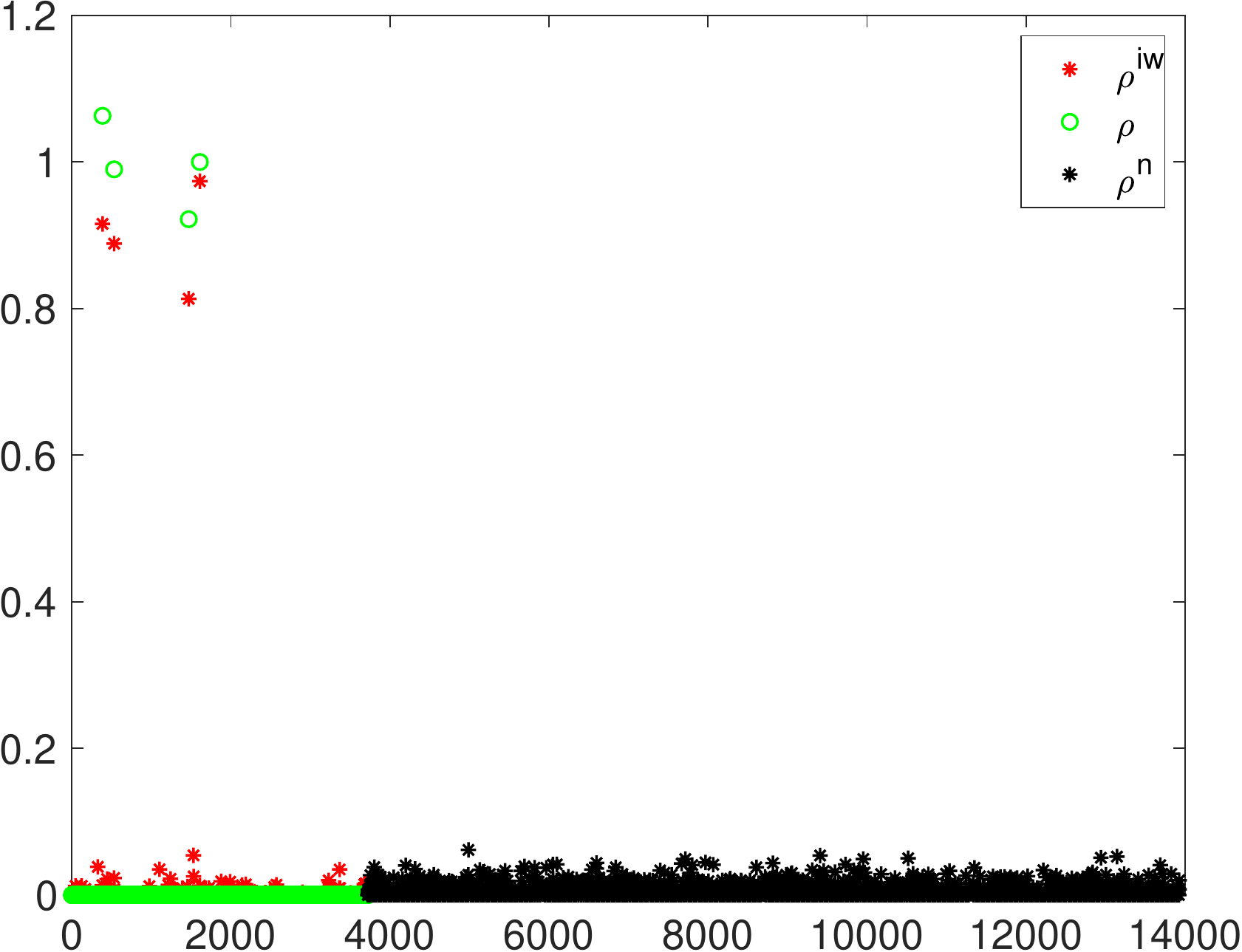}\\
\end{tabular}
\end{center}
%\caption{Imaging with corrupted data, SNR~$=4$dB. From left to right: $\vect \rho_{\ell_1}$ (without noise collector),  its comparison (red stars) with the true solution (green circles), %$\vect \rho_{\ell_1}$
% the vector $\vect \bfrho_{iw}\in\mathbb{C}^{K}$  (with noise collector) and its comparison  (red stars)  with the true solution (green circles). The vector $\vect \bfrho_{noise}\in\mathbb{C}^{\Sigma}$  that absorbs the noise is plotted with black stars. The first $K=3721$ components of the solution corresponding to the physical unknown are plotted with red stars and the $\Sigma=12000$ next components that correspond to the noise collector are plotted with black stars.}
 \caption{Imaging with noisy data, SNR~$=4$dB. The top and the bottom rows show the results without and with the noise collector, respectively. 
 The left columns show the  $\vect \rho_{\ell_1}$ images (the true solution is displayed with white crosses) and the right column the comparison (red stars) with the true solution (green circles).
 In the bottom right image, the first $K=3721$ components of the solution corresponding to the IW are plotted with red stars, and the $\Sigma=12000$ next components corresponding to the noise collector are plotted with black stars.
}
\label{fig:noisec1}
\end{figure}

%\begin{remark}
%For the results shown in two leftmost panels of Figure \ref{fig:noisec1} (plain $\ell_1$ minimization without noise collector), 
%we have plotted the output of the algorithm after a very large number of iterations (30000 iterations). We have verified that more iterations do not improve the solution. 
%\end{remark}

In the next figures, we consider an imaging setup with a large aperture $a/L=1$ and a large
bandwidth $(2B)/\omega_0=1$. Moreover, we increase the pixel size to $\lambda_0$  in both range and cross-range directions, so the Rayleigh resolution is of the order of a pixel. With this imaging  configuration, the columns of the model matrix $\Ac$ are less coherent than in the previous numerical experiments. 
We plot in Figure \ref{fig:noisec2} the $\ell_1$-norm image for a SNR~$=4$dB. With a less coherent matrix $\Ac$ the results are very good. 
%What is different here from the example considered in Figure \ref{fig:noisec1} is the coherence of the columns of the matrix $\Ac$, which is now lower. 
This highlights the inherent difficulty in imaging  when high  resolution is required as in Figure \ref{fig:noisec1} .   
\begin{figure}[htbp]
\begin{center}
\begin{tabular}{cc}
\includegraphics[scale=0.28]{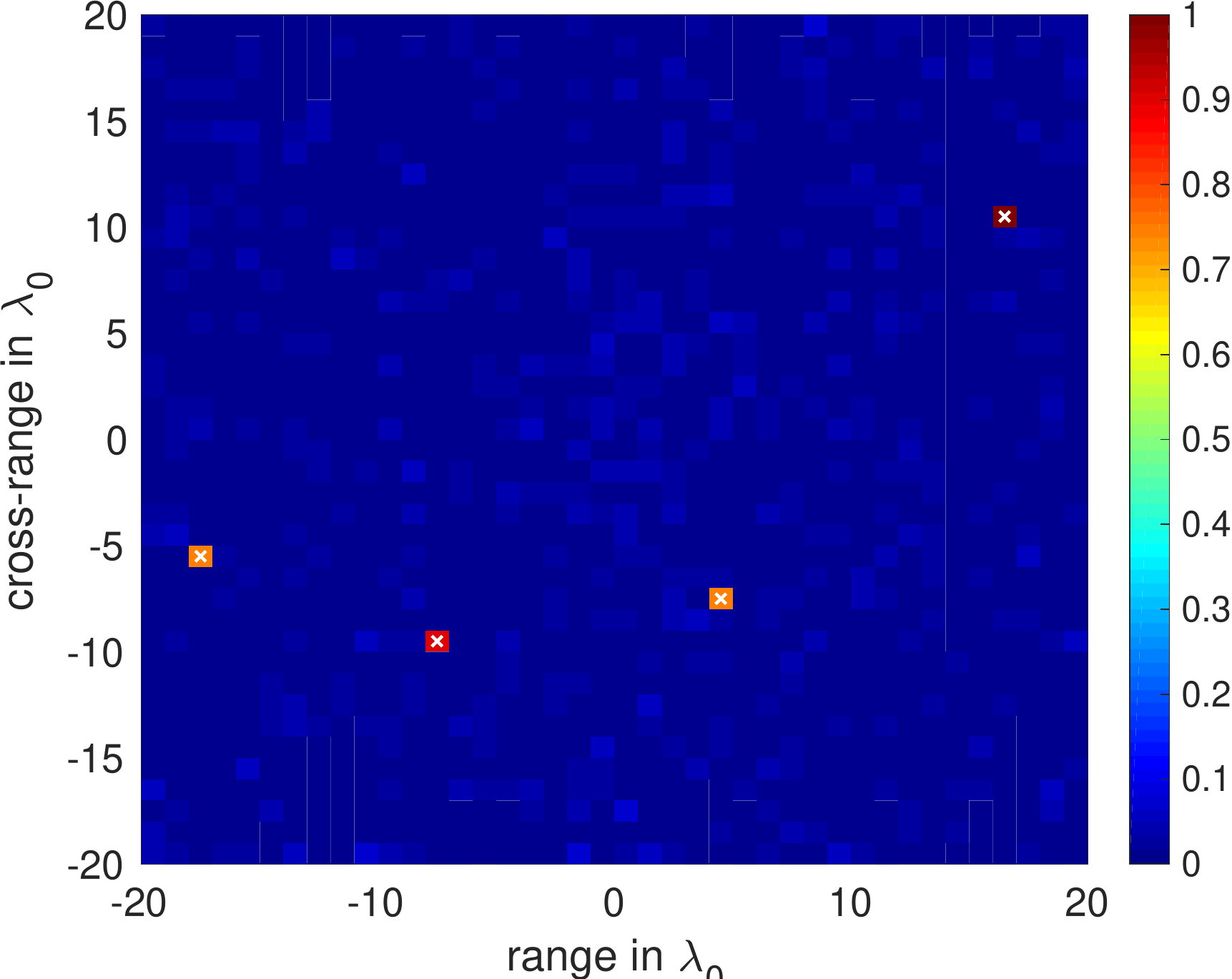}&
\includegraphics[scale=0.28]{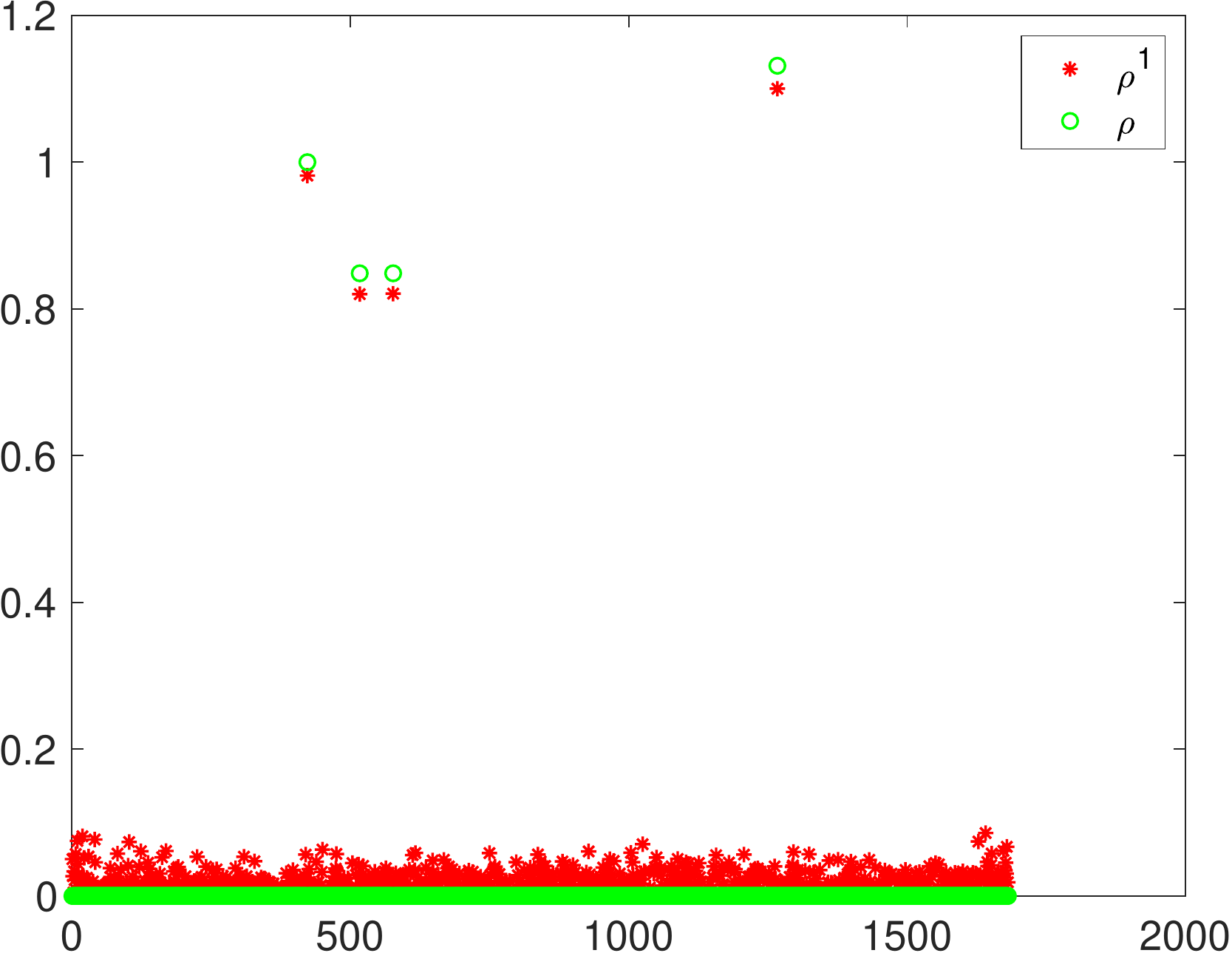}\\
\end{tabular}
\end{center}
\caption{Low resolution images with a moderate level of noise in the data so SNR~$=4$dB. $NS=625$ measurements. $K=1681$ pixels in the images.
%Low resolution imaging ($K=1681$ pixels) with a moderate level of noise so SNR~$=4$dB. The number of measurements is $NS=625$. The solution $\vect \rho_{\ell_1}$ is plotted on the left, and its comparison (red stars) with the true solution (green disks) is plotted on the right. 
}
\label{fig:noisec2}
\end{figure}

For the particular low imaging resolution configuration considered in Figure \ref{fig:noisec2} we obtain good results for a large noise level corresponding to SNR~$=0$dB; see
the top row of Figure \ref{fig:noisec3} where $N\,S=625$ as before. However, when we increase the number of measurements to $NS=1369$, the image obtained with $\ell_1$-norm minimization
turns out to be useless; see the bottom row of Figure \ref{fig:noisec3} . This illustrates the counter-intuitive fact that $\ell_1$-norm minimization does not always benefit from more data, at less if the data is highly contaminated with noise. 
%\textcolor{red}{isn't this statement too strong? should we give a range in which this statement applies? When the number of measurements are already large enough, when the noisy level is high, ...?}. 
This is so because the constant $\gamma$ in (\ref{def:gamma}) depends on the length of the data vector  $\vect b$ as $\sqrt{NS}$. 

%The question is still how much noise  we can add to the data and still get a good image. For this particular low imaging resolution configuration we obtain good results for a relatively large noise level corresponding to SNR~$=0$dB. 
%The minimal $\ell_1$-norm solution is shown in the two leftmost panels of Figure \ref{fig:noisec3}. The two rightmost panels correspond to the same level of noise when we increase the number of
%measurements to $N\, S=1369$ instead of $N\,S=625$ as before. Increasing the data has a dramatic effect on the $\ell_1$-norm image. Adding more data to the $\ell_1$-norm minimization problem (\ref{l1normsol}) completely deteriorates the solution! %In fact, the $\ell_1$-norm minimization algorithm fails to provide a sparse solution, so we plot again the output after 30000 iterations. 

As before, this problem can be fixed with the noise collector as it can be seen in Figure \ref{fig:noisec4}.
%showing the stabilization of the $\ell_1$-norm minimization problem when the noise collector matrix is used. 
Again, the noise is effectively absorbed for both $NS=625$ (top row)  and $NS=1369$ (bottom row) measurements using a matrix collector with a {\em relatively small} number of columns, many less than $e^{NS}$ as Proposition  \ref{noise_c} suggests.
      
\begin{figure}[t]
\begin{center}
\begin{tabular}{cc}
\includegraphics[scale=0.28]{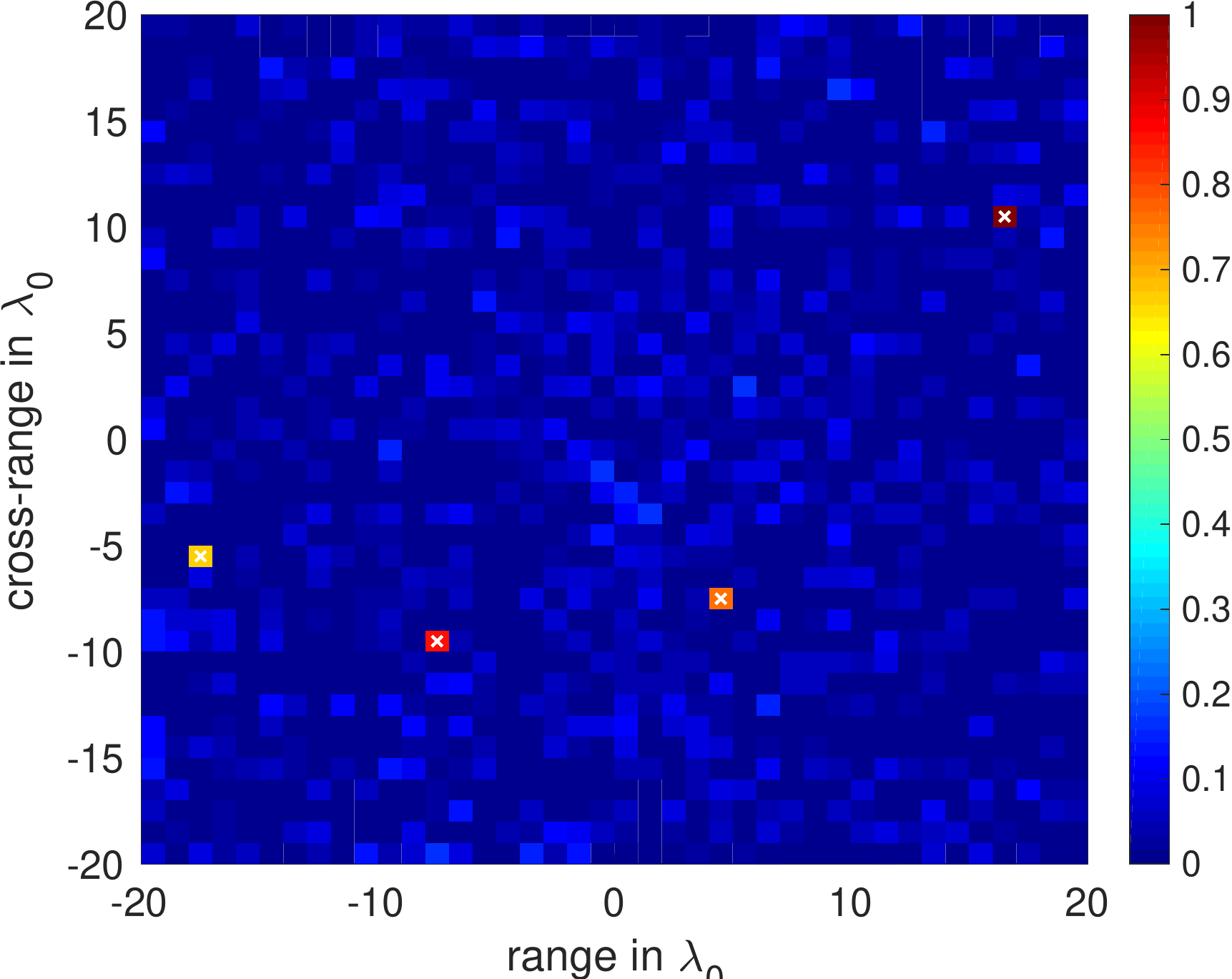}&
\includegraphics[scale=0.28]{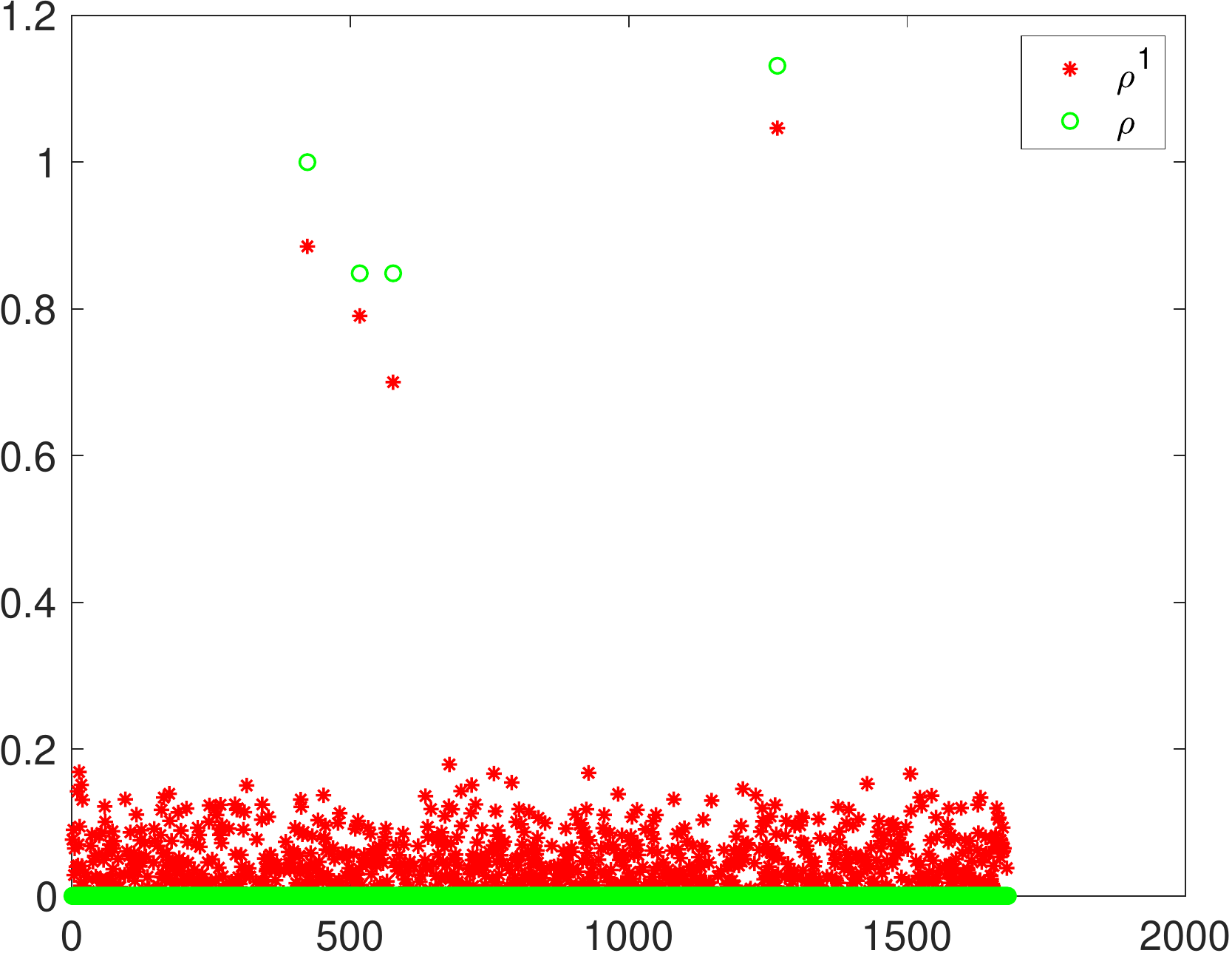}\\
\includegraphics[scale=0.28]{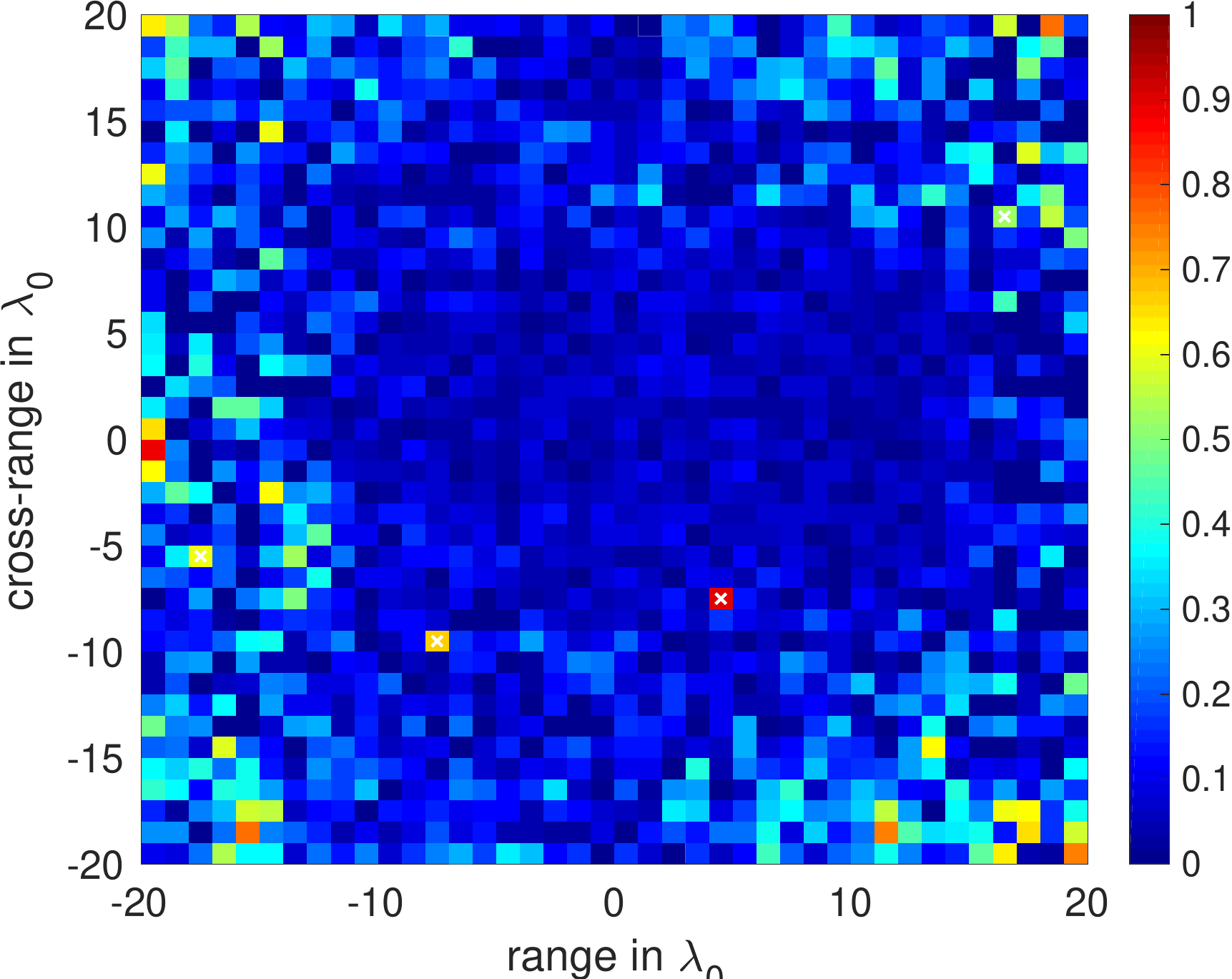}&
\includegraphics[scale=0.28]{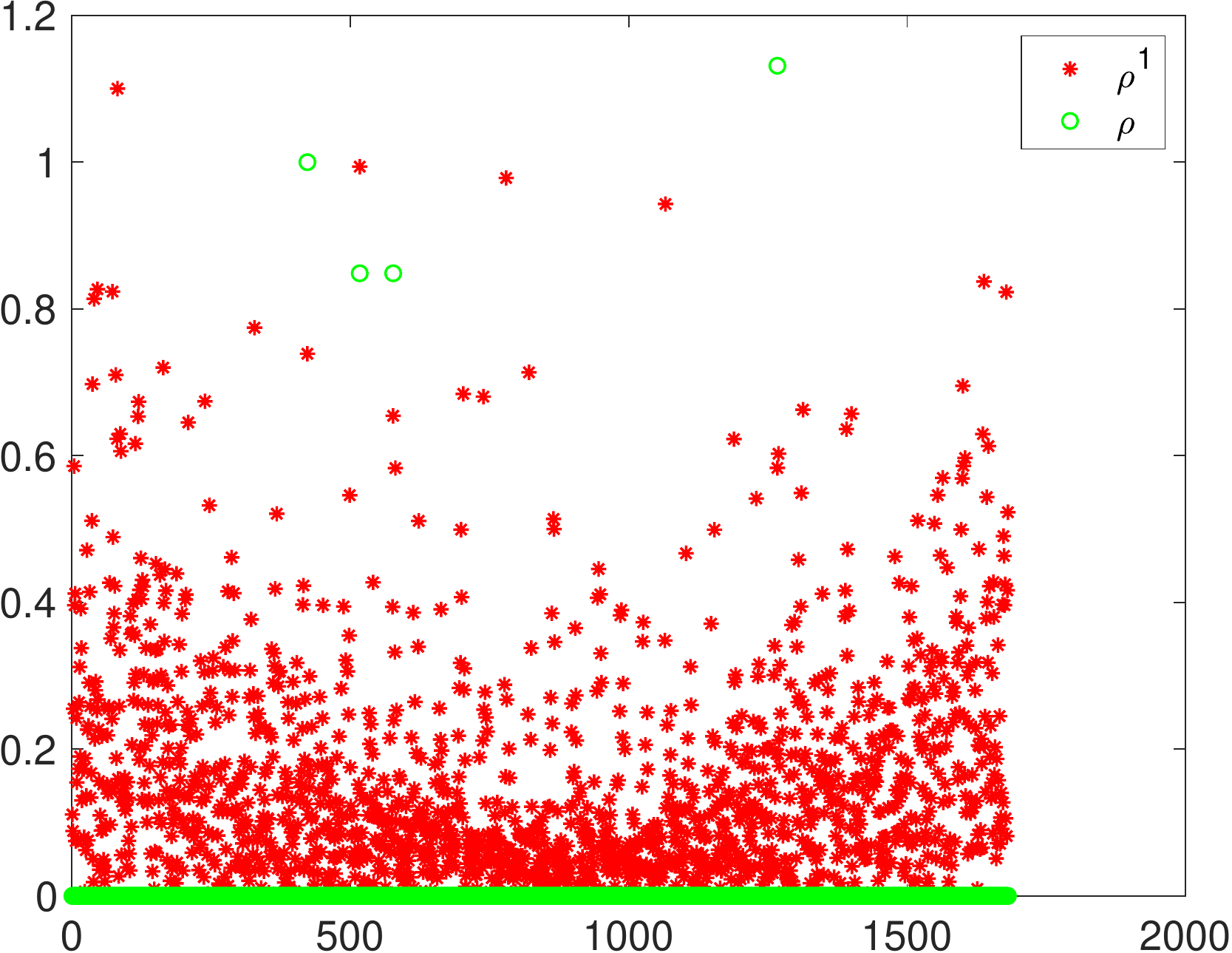}\\
\end{tabular}
\end{center}
\caption{Low resolution images with a high level of noise in the data so SNR~$=0$dB. Top row: $NS=625$ measurements. Bottom row: $NS=1369$ 
measurements. $K=1681$ pixels in the images.}
\label{fig:noisec3}
\end{figure}

\begin{figure}[htbp]
\begin{center}
\begin{tabular}{cc}
\includegraphics[scale=0.28]{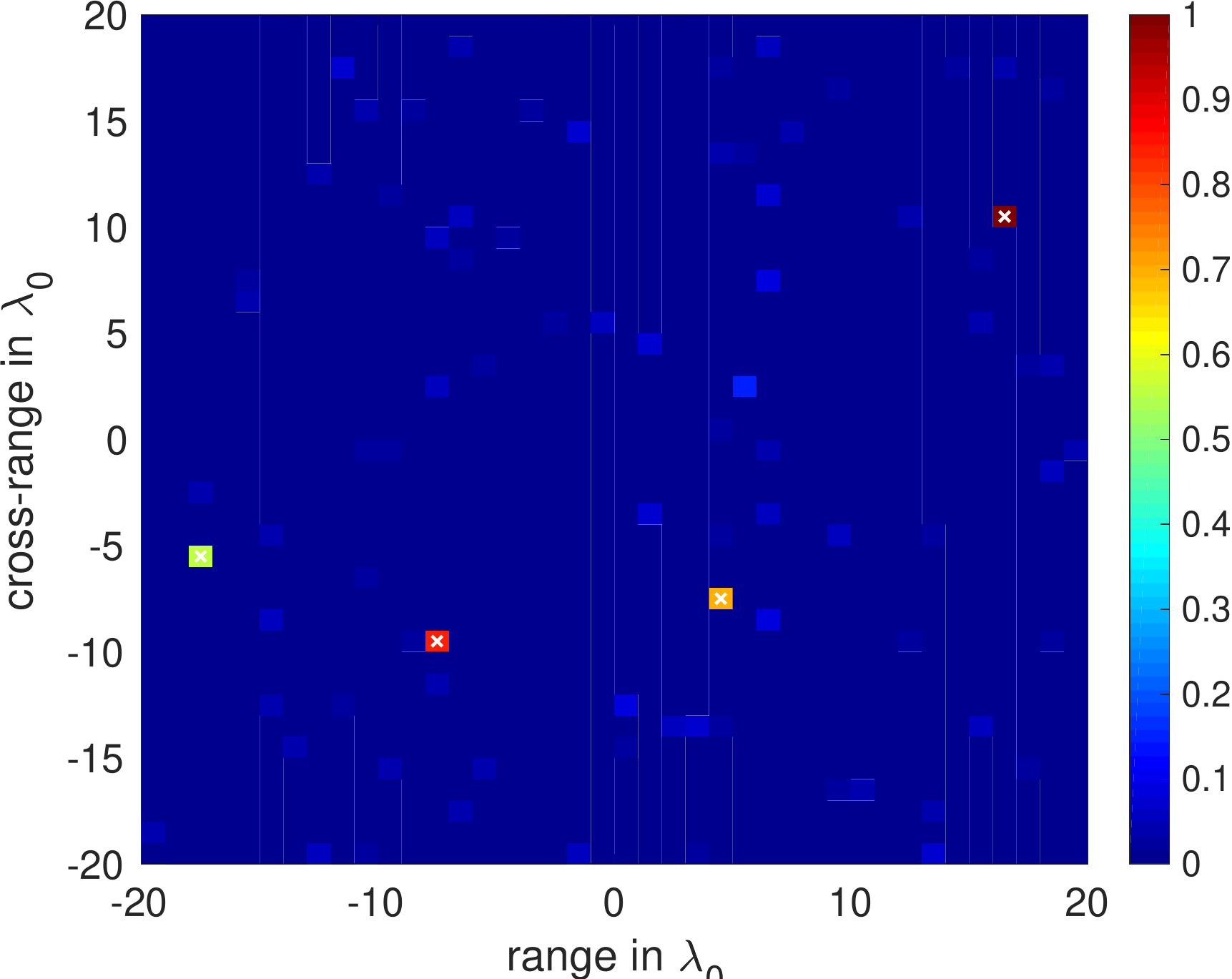}&
\includegraphics[scale=0.28]{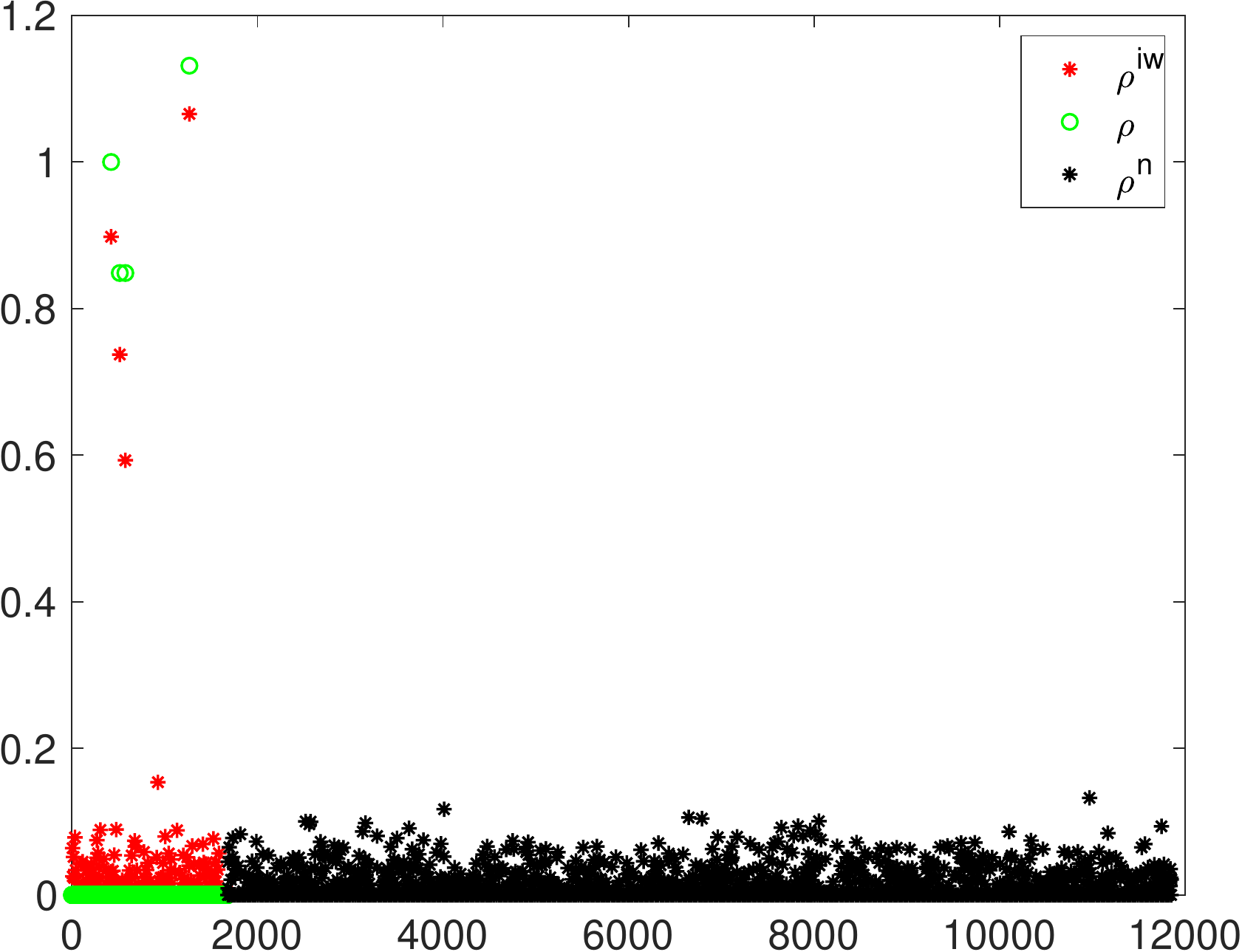}\\
\includegraphics[scale=0.28]{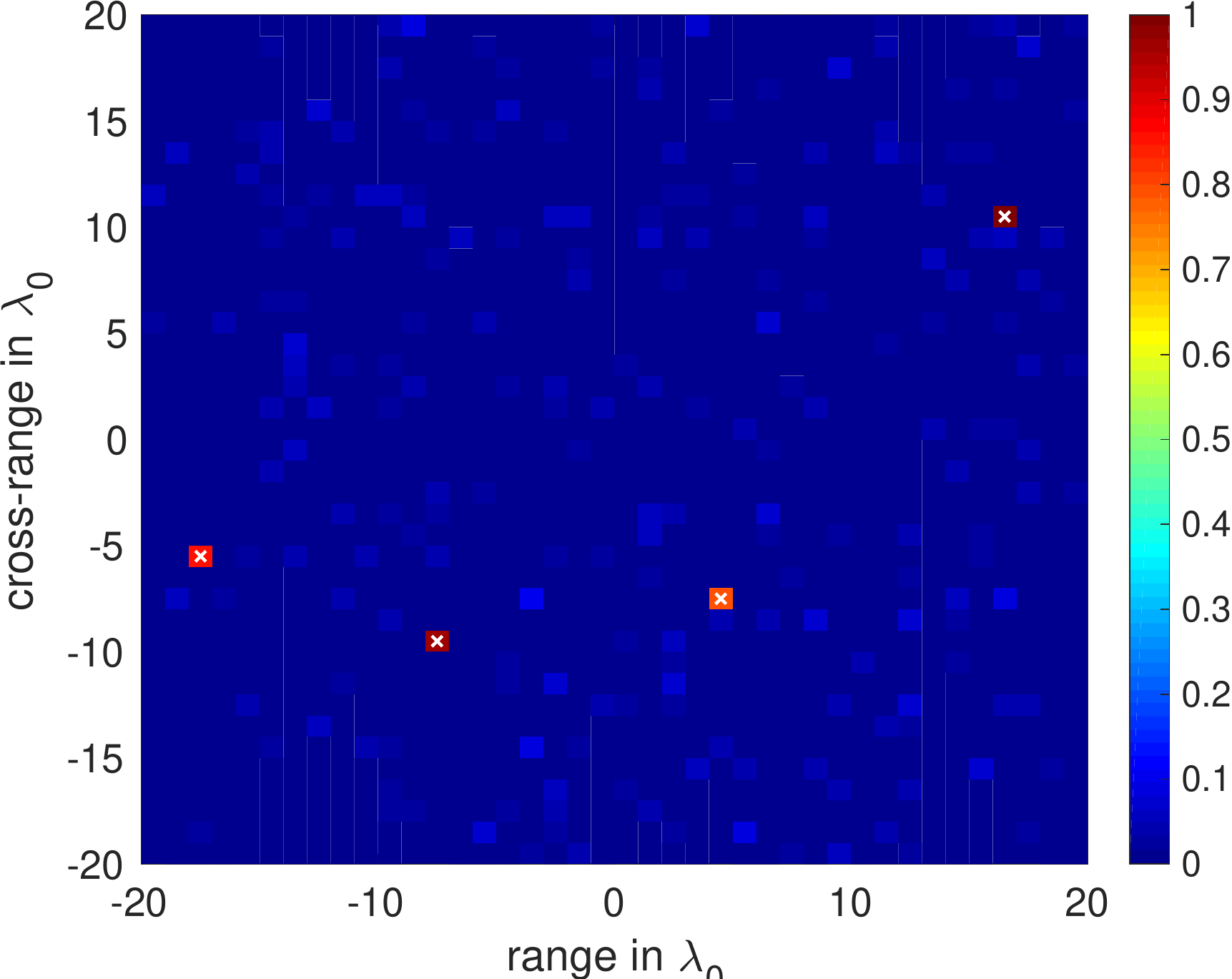}&
\includegraphics[scale=0.28]{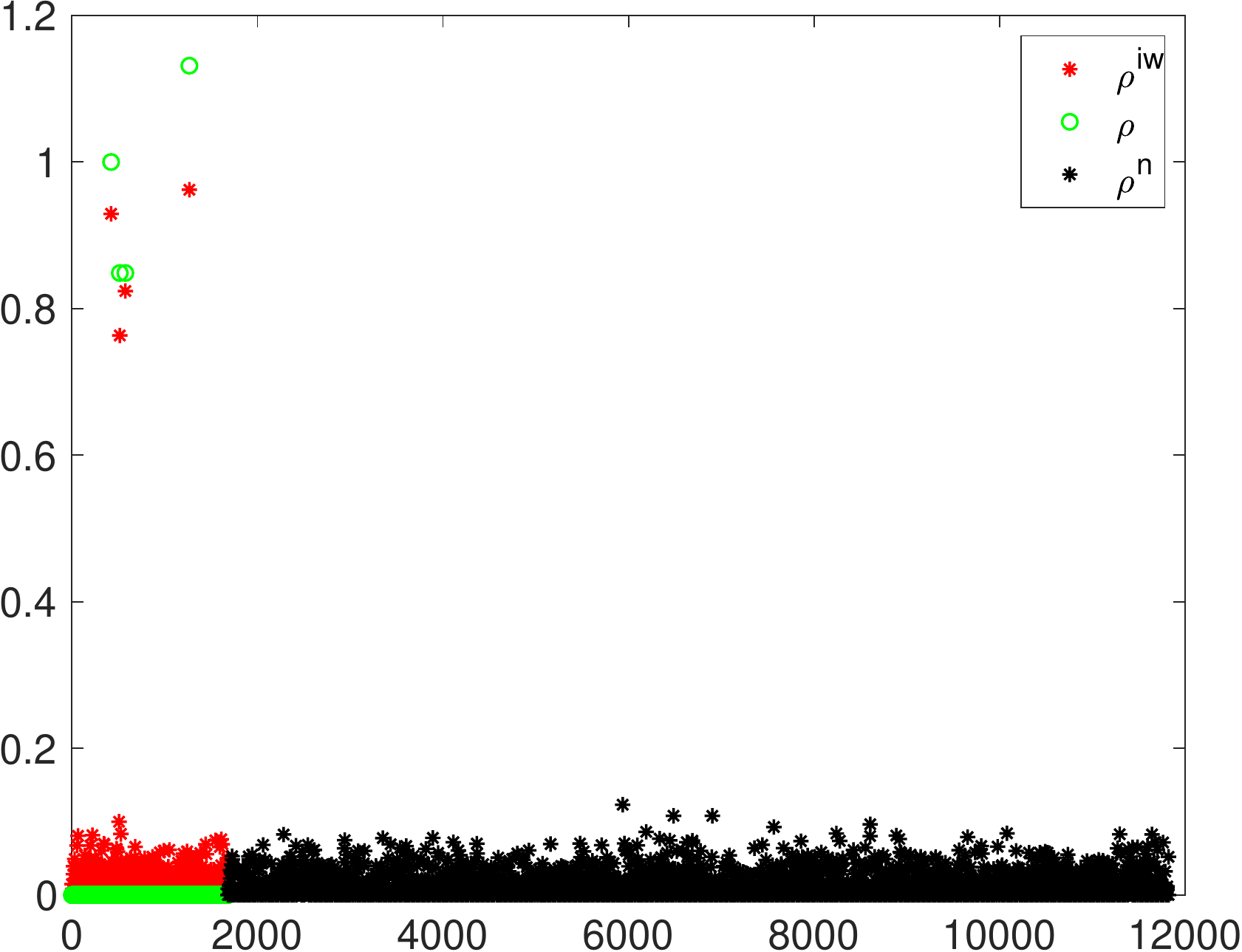}\\
\end{tabular}
\end{center}
\caption{Same as Fig. \ref{fig:noisec3} but with a noise collector matrix $\cC$ with $\Sigma=12000$ columns.}
\label{fig:noisec4}
\end{figure}

\begin{figure}[htbp]
\begin{center}
\begin{tabular}{ccc}
\includegraphics[scale=0.28]{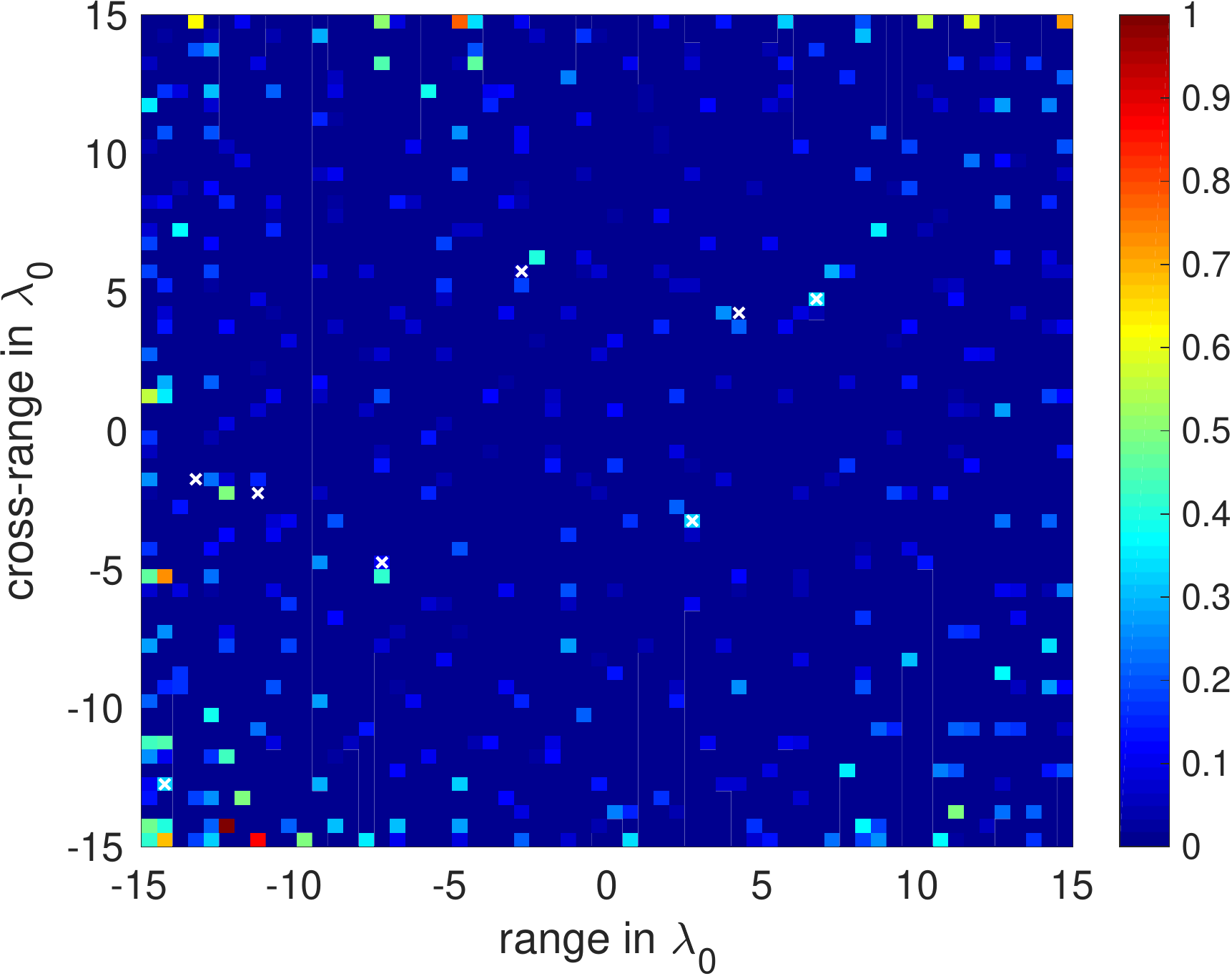}&
\includegraphics[scale=0.28]{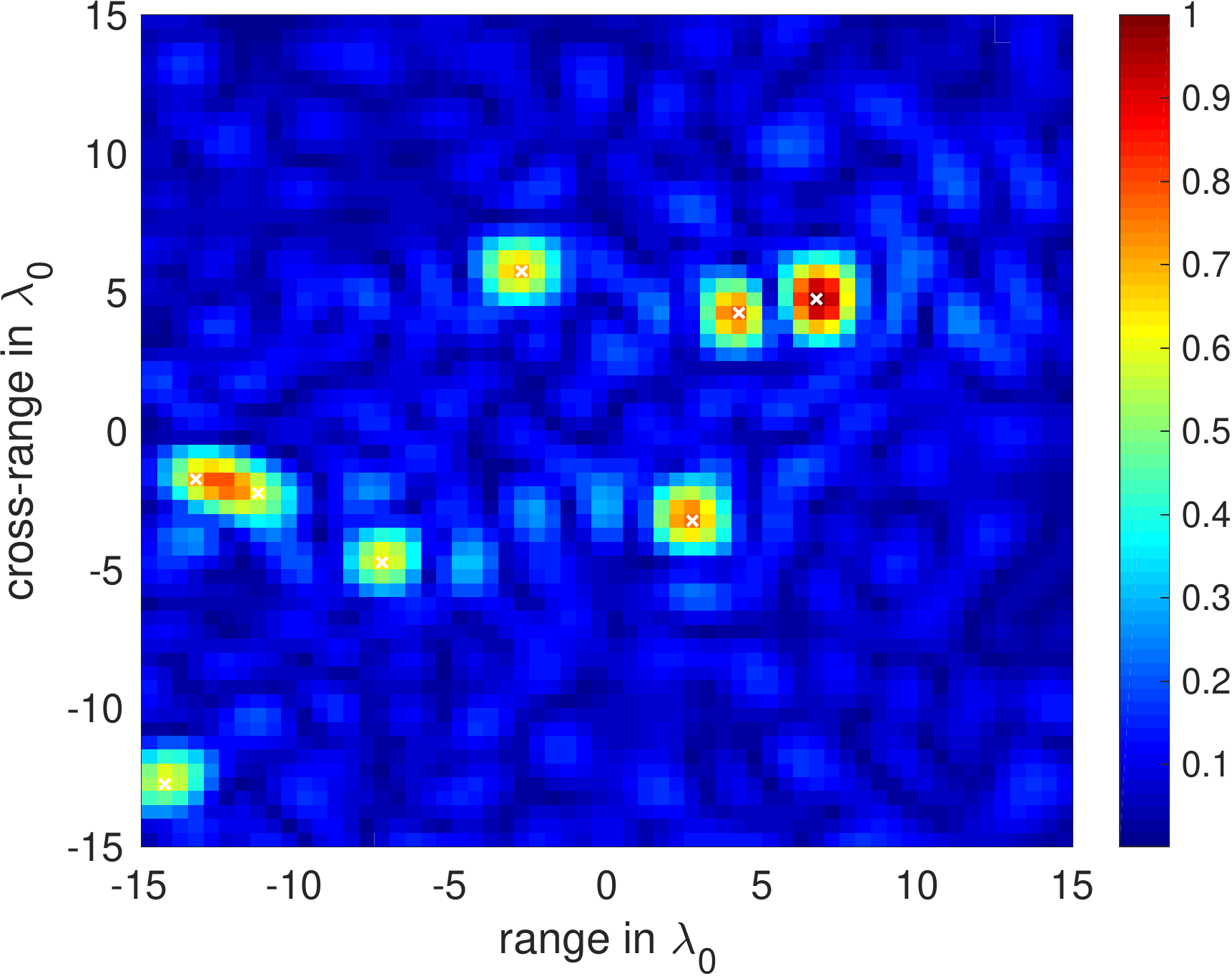}&
\includegraphics[scale=0.28]{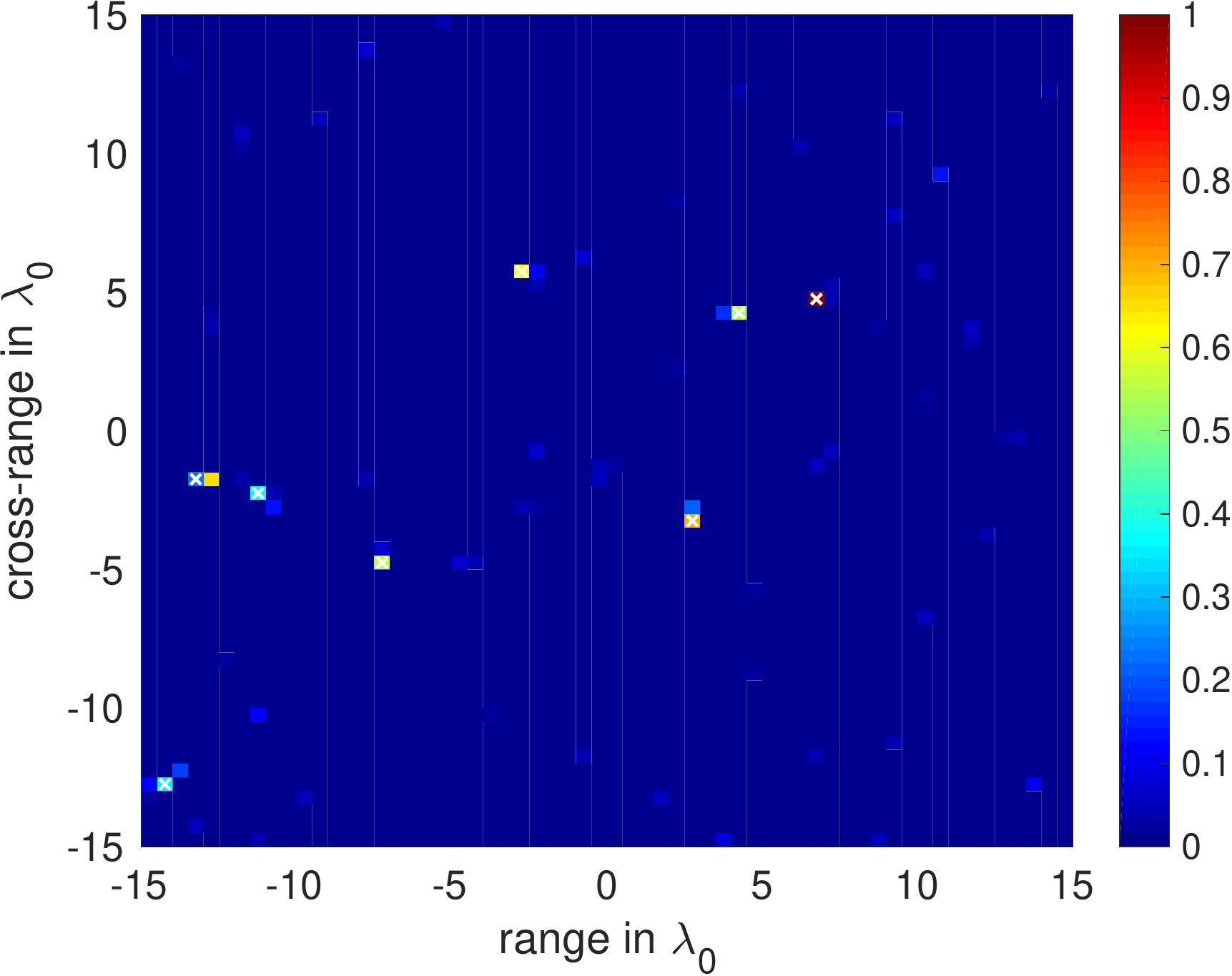}\\
\end{tabular}
\end{center}
\caption{
High resolution images  with a high level of noise in the $NS=625$ data, so SNR~$=0$dB. From left to right: plain $\vect \bfrho_{\ell_1}$ without noise collector,
$\vect \bfrho_{\ell_2}$, and
$\vect \bfrho_{\ell_1}$ using a noise collector. %and its comparison with the true solution shown with green circles (right). 
$K=3721$ pixels in the images.}
\label{fig:noisec5}
\end{figure}

We finish with one last example that shows that the use of the noise collector makes $\ell_1$-norm minimization competitive for imaging sparse scenes because it provides stable results with super-resolution even for highly corrupted data. We consider the example with $M=8$ sources and  SNR~$=0$dB. The array and the bandwidth are relatively large ($a/L=1/2$, $(2B)/\omega_0=1/2$), so the classical $\ell_2$-norm resolution is of the order $O(2\lambda_0)$, as in Figure (\ref{fig:noisec1}).
In Figure \ref{fig:noisec5} we show, from left to right, (i) the minimal $\ell_1$-norm solution without noise collector, which fails to give a good image,
%(there is no convergence), 
(ii) the $\ell_2$-norm solution (\ref{sol_ell2}),  which is stable to additive noise but does not resolve nearby sources, and 
(iii) the minimal  $\ell_1$-norm solution with the noise collector, which provides a very precise and stable image. 
%and its comparison with the true solution.

%**************************************************************************************************************
\section{Discussion}
\label{sec:conclusions}
In this paper, we consider imaging problems that can be formulated as underdetermined linear systems of the form $\Ac \, \rho_\delta =\bfb_\delta$, where $\Ac$ is an $N\times K$ model matrix with $N \ll K$, and  $\bfb_\delta$ is the $N$-dimensional data vector contaminated with noise. We assume that the solution is an M-sparse vector in $\mC^{K}$, corresponding to the $K$ pixels of the IW. 
We consider additive  noise in the data, so the data vector can be decomposed as $\bfb_\delta= \bfb + \delta \bfb$, where $\bfb$ is the data vector in the absence of noise and $\delta \bfb$ is 
the noise vector. 
We provide a theoretical framework that allows us to examine under what conditions the $\ell_1$-minimization problem admits a solution that is close to the exact one. We also
shown that, for our imaging problems,  $\ell_1$-minimization fails when the noise level is high and the dimension $N$ of the data vector $\bfb_\delta$ increases. The reason is that 
the error is proportional to the  square root of $N$. 

To alleviate this problem and increase the robustness of $\ell_1$-minimization, we propose a regularization strategy.
%that consists in increasing the \textcolor{blue}{column space} of the matrix $\Ac$. 
In particular, we seek the solution of $[\Ac \, | \,  \Cc] \, \bfrho_\delta=\bfb_\delta$, where the $N \times \Sigma$ matrix $\Cc$ is a noise collector. Thus,
the unknown $\rho_\delta$ is now a 
vector in $\mC^{K+\Sigma}$. The first $K$ components of the unknown correspond to the distribution of sources in the IW, while the $\Sigma$ next components do not correspond to any physical quantity. They are introduced to provide a fictitious source distribution given by an appropriate linear combination of the columns of $\Cc$ that 
produces a good approximation to $\delta \bfb$. The main idea is to create a library of noises. 
The columns of the noise collector matrix  are elements of this library and they are constructed so as to be incoherent with respect to the columns of $\Ac$. 
Theoretically, the dimension $\Sigma$ of the noise collector increases exponentially with $N$,
 which suggests that it may not be useful in practice. Our numerical results show, however,  robustness for 
$\ell_1$-minimization in the presence of noise  when a large enough number of columns $\Sigma \precsim  10 K$ is used to build the noise collector matrix. 

Our first findings on the noise collector are very encouraging. We have shown that its use improves dramatically the robustness of $\ell_1$-norm reconstructions 
when the data are corrupted with additive uncorrelated noise. Many other questions ought to be addressed. Some directions of our future research concern 
the following aspects: what happens with other types of noise?,
can  we design noise collectors adaptively depending on the noise in the data?, what if the noise comes from wave propagation in a random medium?,
can we design a noise collector for this case?, how much do we need to know about the noise so as to design a good noise collector?, can we retrieve 
this information from the data? Some of these questions will be addressed somewhere else.

\section*{Acknowledgments}
Part of this material is based upon work supported by the National Science Foundation under Grant No. DMS-1439786 while the authors were in 
residence at the Institute for Computational and Experimental Research in Mathematics (ICERM) in Providence, RI, during the Fall 2017 semester.
The work of M. Moscoso was partially supported by Spanish MICINN grant FIS2016-77892-R. The work of A.Novikov was partially supported by NSF grants  DMS-1515187, DMS-1813943.
The work of C. Tsogka was partially supported by AFOSR FA9550-17-1-0238.

\appendix

\section{Proof of Proposition~\ref{Old_l1}}\label{uno}

We will now prove auxiliary lemmas that we will use in the proof of Propostion~\ref{Old_l1}.

\begin{lemma}~\label{on_e_v}
Let $\Bc$ be an $M \times M$ Hermitian matrix such that $b_{ii}=1$, and $|b_{ij}| \leq c$ for all $i \neq j$. Assume $(M-1)c <1$, then
any eigenvalue $\lambda$ of $\Bc$ satisfies
\begin{equation}\label{eig_v}
1-(M-1)c \leq \lambda \leq 1+(M-1)c.
\end{equation}
\end{lemma}

\begin{proof}
Suppose $\Bc \vect \rho = \lambda \vect \rho$. By the triangle inequality for any row $i$ we have 
\[
|\rho_i| - \left| \sum_{j \neq i} b_{i j} \rho_j \right| \leq |\lambda  \rho_i | \leq |\rho_i| + \left| \sum_{j \neq i} b_{i j} \rho_j  \right|. 
\]
Since 
$\left| \sum_{j \neq i} b_{i j} \rho_j \right|  \leq (M-1) c$, we obtain~(\ref{eig_v}). 

\end{proof}

\begin{lemma}\label{l1l1} 
Suppose $\gamma$ is defined by~(\ref{def:gamma}). Let  $\vect \rho_1$ and   $\vect \rho_2$ be minimizers of
$\|\vect \eta\|_{\ell_1}$,  subject to $\Ac \vect \eta= \vect b_1$ and $\Ac \vect \eta= \vect b_2$, respectively.
Then, there exists $\xi$ such that  $\Ac \vect \xi= \vect b_1$, 
\begin{equation}\label{xi1}
 \| \vect \xi \|_{\ell_1}  \leq   \| \vect \rho_1 \|_{\ell_1}  +  2 \gamma  \| \vect b_1 - \vect b_2  \|_{\ell_2},
\end{equation}
and 
\begin{equation}\label{l1_l1_a}
    \| \vect \xi - \vect \rho_2 \|_{\ell_1}  \leq  \gamma \| \vect b_1 - \vect b_2  \|_{\ell_2}.
    \end{equation}
\end{lemma}

\begin{proof}
Let us first show that
\begin{equation}\label{l1_l1}
| \| \vect \rho_1 \|_{\ell_1} - \| \vect \rho_2 \|_{\ell_1} | \leq \gamma \| \vect b_1 - \vect b_2  \|_{\ell_2}. 
\end{equation}
Assume, for definiteness, that $\| \vect \rho_1 \|_{\ell_1} > \| \vect \rho_2 \|_{\ell_1}$. Then,
\[
| \| \vect \rho_1 \|_{\ell_1} - \| \vect \rho_2 \|_{\ell_1} | =  \| \vect \rho_1 \|_{\ell_1} - \| \vect \rho_2 \|_{\ell_1}.
\]
Suppose $\vect \rho_3$ is a minimizer of
$\|\vect \eta\|_{\ell_1}$,  subject to $\Ac \vect \eta= \vect b_1 - \vect b_2$. Since 
$\Ac (\vect \rho_2 + \vect \rho_3) = \vect b_1 $, and $\vect \rho_1$ is a minimizer of
$\|\vect \eta\|_{\ell_1}$,  subject to $\Ac \vect \eta= \vect b_1$, it follows
$ \| \vect \rho_1 \|_{\ell_1}  \leq \| \vect \rho_2 + \vect \rho_3 \|_{\ell_1}$. By~(\ref{def:gamma}) and the triangle inequality
\[
 \| \vect \rho_2 + \vect \rho_3 \|_{\ell_1} \leq  \| \vect \rho_2 \|_{\ell_1} +  \gamma \| \vect b_1 - \vect b_2  \|_{\ell_2}. 
\]
Thus, ~(\ref{l1_l1}) holds.

Let $\xi =\vect \rho_2 + \vect \rho_3$, where $\vect \rho_3$ is a minimizer of
$\|\vect \eta\|_{\ell_1}$,  subject to $\Ac \vect \eta= \vect b_1 - \vect b_2$.  Then, $\Ac \vect \xi= \vect b_1$, and inequality~(\ref{l1_l1_a}) follows from~(\ref{def:gamma}).
Using~(\ref{l1_l1}),~(\ref{l1_l1_a}), and the triangle inequality we obtain
\[
\| \vect \xi \|_{\ell_1}  \leq  \| \vect \rho_2  \|_{\ell_1}+ \| \vect \xi - \vect \rho_2 \|_{\ell_1} \leq  \| \vect \rho_1 \|_{\ell_1}  +  2 \gamma  \| \vect b_1 - \vect b_2  \|_{\ell_2}.
\]
\end{proof}

\begin{lemma}\label{l1l1_a}
Suppose $\Ac \vect \rho = \Ac \vect \xi = \vect b$, where $\vect \rho$ is $M$-sparse, and $\vect \xi$ is arbitrary. Assume vicinities (\ref{vicinity}) do not overlap. Then,
\begin{equation}\label{co_inn}
 {\bf Co}(\vect \rho, \vect \xi)  \leq \frac{1}{2}  {\bf In}(\vect \rho, \vect \xi).
\end{equation}
In particular,
\begin{equation}\label{exact0}
 \| \vect \rho \|_{\ell_1} \leq  \| \vect \xi \|_{\ell_1}. 
\end{equation}
\end{lemma}

\begin{proof}
For any $\mu \in \mathbb{C}^M$, we have
\[
0=\langle \Ac_T  \left( \Ac^*_T \Ac_T \right)^{-1} \vect \mu,  \Ac ( \vect \rho - \vect  \xi) \rangle = \langle  \vect \mu ,    \left( \Ac^*_T \Ac_T \right)^{-1} \Ac^*_T \Ac ( \vect \rho - \vect  \xi) \rangle
\]
since $0= \Ac ( \vect \rho - \vect  \xi)$.
By Lemma~\ref{on_e_v}, the largest eigenvalue of $\left(\Ac^*_T \Ac_T \right)^{-1}$ is smaller than $3/2$. Thus,
\[
\left| \sum_{j \in T}  \bar{\mu}_j \rho_j - \sum_{j \in T} \sum_{k \in S_j}  \langle \vect a_j, \vect a_k \rangle \bar{\mu}_j \xi_k \right| \leq \frac{3}{2}
 \sum_{j \in T} 
\sum_{k  \not\in \Upsilon} \left| \langle \vect a_j, \vect a_k \rangle \bar{\mu}_j \xi_k \right|, \, \Upsilon = \cup_{j \in T} S_j.
\]
Choose $ \mu_j$,  so that $|\mu_j|=1$ and
\[
\left| \sum_{j \in T}  \bar{\mu}_j \rho_j - \sum_{j \in T} \sum_{k \in S_j}  \langle \vect a_j, \vect a_k \rangle \bar{\mu}_j \xi_k \right| ={\bf Co}(\vect \rho, \vect \xi).
\]
We can estimate
\[
{\bf Co}(\vect \rho, \vect \xi) \leq \frac{3}{2}  \frac{1}{3 M}   \sum_{j \in T} \sum_{k  \not\in \Upsilon} \left| \xi_k \right|  \leq \frac{1}{2}  \sum_{k  \not\in \Upsilon} \left| \xi_k \right|. 
\]
which is equivalent to~(\ref{co_inn}). Observe that (see~(\ref{co_}))
\[
 \| \vect \rho \|_{\ell_1} -  \sum_{k  \in \Upsilon} \left| \xi_k \right| \leq {\bf Co}(\vect \rho, \vect \xi). 
\]

\end{proof}

\begin{proof}[Proof of Proposition~\ref{Old_l1}]
If  $\vect \rho$ and   $\vect \rho_\delta$ are minimizers of
$\|\vect \eta\|_{\ell_1}$,  subject to $\Ac \vect \eta= \vect b$ and $\Ac \vect \eta= \vect b_\delta$, respectively, we can apply
 Lemma~\ref{l1l1}  and conclude
there exists  $\xi$ such that  $\Ac \vect \xi= \vect b$, 
\begin{equation}\label{xi1_proof_}
 \| \vect \xi \|_{\ell_1}  \leq   \| \vect \rho \|_{\ell_1}  + 2 \gamma \delta,
\end{equation}
and 
\begin{equation}\label{l1_l1_a_proof_}
    \| \vect \xi - \vect \rho_\delta \|_{\ell_1}  \leq \gamma \delta.
    \end{equation}
Since
 \[
 \| \vect \rho \|_{\ell_1}  \leq  {\bf Co}(\vect \rho, \vect \xi)   + \sum_{j \in T}  |\xi_j |,
 \]
  by Lemma~\ref{l1l1_a} we have 
\begin{equation}\label{co_in_}
  \| \vect \rho \|_{\ell_1} \leq  \frac{1}{2}  \sum_{j  \not\in T} \left| \xi_j \right| +  \sum_{j \in T}  |\xi_j | = \| \vect{\xi} \|_{\ell_1} -  \frac{1}{2}  {\bf In}(\vect \rho, \vect \xi).
 \end{equation}

 Comparing~(\ref{co_in_}) and~(\ref{xi1_proof_}) we conclude
\begin{equation}\label{in_easy_}
{\bf In}(\vect \rho, \vect \xi) \leq   4 \gamma \delta.
\end{equation}
By the triangle inequality and~(\ref{l1_l1_a_proof_}), we have 
\[
{\bf In}(\vect \rho, \vect \rho_\delta) \leq {\bf In}(\vect \rho, \vect \xi) +    \| \vect \xi - \vect \rho_\delta \|_{\ell_1}  \leq  5
\gamma \delta.
\]
 Hence, we have obtained~(\ref{est_in-}). From~(\ref{co_inn})  and~(\ref{in_easy_}), we obtain
 \[
 {\bf Co}(\vect \rho, \vect \xi)  \leq 2 \gamma \delta.
 \]
By the triangle inequality and~(\ref{l1_l1_a_proof_}), we have 
\[
{\bf Co}(\vect \rho, \vect \rho_\delta) \leq {\bf Co}(\vect \rho, \vect \xi) +    \| \vect \xi - \vect \rho_\delta \|_{\ell_1}  \leq 
 3 \gamma \delta.
\]
If the noise level $\delta=0$, then ${\bf Co}(\vect \rho, \vect \xi) ={\bf In}(\vect \rho, \vect \xi) =0$. It means
$\supp(\vect \rho_\delta) \subset \Upsilon.$ Since $\Ac \vect \rho_\delta = \Ac \vect \rho$, we can use~(\ref{exact0}). 
Note that the inequality~(\ref{exact0}) 
becomes strict  if  
$\Upsilon$  does not contain collinear vectors. Thus, we conclude
 $\vect \rho_\delta = \vect \rho$.
\end{proof}

 \end{document}